\documentclass[12pt,technote, onecolumn, draft]{IEEEtran}
\usepackage{latexsym}
\usepackage{mathrsfs}
\usepackage{multirow}
\usepackage{amsfonts,amssymb,amsmath,amsthm,bm}
\usepackage{color}
\usepackage{cite}
\newtheorem{theorem}{Theorem}
\newtheorem{example}{Example}

\newtheorem{remark}{Remark}
\newtheorem{corollary}{Corollary}

\newtheorem{proposition}{Proposition}

\begin{document}

\title{Further results on bent partitions$^{\dag}$}
\author{Jiaxin Wang, Yadi Wei, Fang-Wei Fu
\IEEEcompsocitemizethanks{\IEEEcompsocthanksitem Jiaxin Wang is with the School of Mathematics, Hefei University of Technology, Hefei 230601, China, Email: wjiaxin@hfut.edu.cn, Yadi Wei and Fang-Wei Fu are with the Chern Institute of Mathematics and LPMC, Nankai University, Tianjin 300071, China, Emails: wydecho@mail.nankai.edu.cn, fwfu@nankai.edu.cn.
}
\thanks{$^\dag$This research is supported by the National Key Research and Development Program of China (Grant No. 2022YFA1005000), the National Natural Science Foundation of China (Grant Nos. 12141108, 62371259,12411540221), the Fundamental Research Funds for the Central Universities of China (Nankai University), and the Nankai Zhide Foundation.}
\thanks{manuscript submitted  September 20, 2025}
}

\maketitle

\begin{abstract}
  Bent partitions of $V_{n}^{(p)}$ play an important role in constructing (vectorial) bent functions, partial difference sets, and association schemes, where $V_{n}^{(p)}$ denotes an $n$-dimensional vector space over the finite field $\mathbb{F}_{p}$, $n$ is an even positive integer, and $p$ is a prime. For bent partitions, there remains a challenging open problem: Whether the depth of any bent partition of $V_{n}^{(p)}$ is always a power of $p$. Notably, the depths of all current known bent partitions of $V_{n}^{(p)}$ are powers of $p$. In this paper, we prove that for a bent partition $\Gamma$ of $V_{n}^{(p)}$ for which all the $p$-ary bent functions generated by $\Gamma$ are regular or all are weakly regular but not regular, the depth of $\Gamma$ must be a power of $p$. We present new constructions of bent partitions that (do not) correspond to vectorial dual-bent functions. In particular, a new construction of vectorial dual-bent functions is provided. Additionally, for general bent partitions of $V_{n}^{(2)}$, we establish a characterization in terms of Hadamard matrices.
\end{abstract}

\begin{IEEEkeywords}
Bent partitions; depth of bent partitions; (vectorial) bent functions; vectorial dual-bent functions; Hadamard matrices
\end{IEEEkeywords}

\section{Introduction}
Let $V_{n}^{(p)}$ be an $n$-dimensional vector space over the finite field $\mathbb{F}_{p}$, where $p$ is a prime. In \cite{AM2022Be}, Anbar and Meidl introduced bent partitions, which play a significant role in constructing (vectorial) bent functions \cite{AKM2023On,AM2022Be}, partial difference sets \cite{AKM2022Be,WFW2023Be} and association schemes \cite{AKM2024A}. For a bent partition $\Gamma$ of $V_{n}^{(p)}$ of depth $K$, there is a challenging open problem proposed in \cite{AM2022Be}:
\begin{itemize}
  \item Is $K$ always a power of $p$?
\end{itemize}
The depth of all known bent partitions \cite{AM2022Be,AKM2023Ge,MP2021Be,WFW2023Be} are prime powers. The preimage partition of any ternary bent function is a bent partition of depth $3$ (see \cite{AAKM2024Be}). The preimage partition of any vectorial bent function from $V_{n}^{(2)}$ to $V_{2}^{(2)}$ is a bent partition of depth $4$ (see \cite{AFKMWW2025A}). For the depth $K>4$, the bent partitions in \cite{AM2022Be,AKM2023Ge,MP2021Be,WFW2023Be} all correspond to vectorial dual-bent functions. In \cite{AFKMWW2025A}, the first explicit construction of bent partition of depth $K>4$ not corresponding to vectorial dual-bent functions was given. It is interesting to construct more bent partitions that do not correspond to vectorial dual-bent functions.

Motivated by the above research questions, we further study bent partitions. In this paper, we consider bent partitions $\Gamma$ of $V_{n}^{(p)}$ for which all the $p$-ary bent functions generated by $\Gamma$ are regular or all are weakly regular but not regular (such bent partitions are called belonging to class $\mathcal{WBP}$). We give a sufficient and necessary condition for the preimage set partition of a weakly regular $p$-ary bent function $f: V_{n}^{(p)}\rightarrow \mathbb{F}_{p}$ to be a bent partition (where $n$ is even and $p$ is odd). Building on this condition, we then prove that for any bent partition $\Gamma$ of $V_{n}^{(p)}$ belonging to class $\mathcal{WBP}$, the depth $K$ of $\Gamma$ must be a power of $p$. Particularly, the depth of any bent partition $\Gamma$ of $V_{n}^{(2)}$ must be a power of $2$. We give new constructions of bent partitions which (do not) correspond to vectorial dual-bent functions. In particular, a new construction of vectorial dual-bent functions is presented. In \cite{WFWY2024A}, for bent partitions of $V_{n}^{(2)}$ corresponding to vectorial dual-bent functions, a characterization in terms of Hadamard matrices was given. In this paper, for general bent partitions of $V_{n}^{(2)}$, we provide a characterization in terms of Hadamard matrices.

The rest of the paper is organized as follows. Section II provides the needed preliminaries. In Section III, we show that the depth of certain bent partitions must be prime power. In Section IV, we give new constructions of bent partitions (not) corresponding to vectorial dual-bent functions, and characterize general bent partitions of $V_{n}^{(2)}$ in terms of Hadamard matrices. In Section V, a conclusion is given.

\section{Preliminaries}
In this section, we give the needed results on (vectorial) bent functions, bent partitions, and generalized Hadamard matrices. First, some notations are given used throughout this paper.

\begin{itemize}
  \item $p$ is a prime.
  \item $\zeta_{p}$ is a complex primitive $p$-th root of unity.
  \item $\mathbb{F}_{p^n}$ is the finite field with $p^n$ elements.
  \item $\mathbb{F}_{p}^{n}$ is the vector space of the $n$-tuples over $\mathbb{F}_{p}$.
  \item $V_{n}^{(p)}$ is an $n$-dimensional vector space over $\mathbb{F}_{p}$.
  \item $\langle, \rangle_{n}$ denotes a (non-degenerate) inner product of $V_{n}^{(p)}$. When $V_{n}^{(p)}=\mathbb{F}_{p^n}$, let $\langle a, b\rangle_{n}=Tr_{1}^{n}(ab)$, where $Tr_{m}^{n}$ denotes the trace function from $\mathbb{F}_{p^n}$ to $\mathbb{F}_{p^m}$, $m \mid n$; when $V_{n}^{(p)}=\mathbb{F}_{p}^{n}$, let $\langle a, b\rangle_{n}=a \cdot b$, where $\cdot$ is the dot product; when $V_{n}^{(p)}=V_{n_{1}}^{(p)} \times \dots \times V_{n_{s}}^{(p)}$, let $\langle a, b\rangle_{n}=\sum_{i=1}^{s}\langle a_{i}, b_{i}\rangle_{n_{i}}$, where $a=(a_{1}, \dots, a_{s}), b=(b_{1}, \dots, b_{s}) \in V_{n}^{(p)}$.
  \item For a function $F: S_{1}\rightarrow S_{2}$, let $D_{F, i}=\{x \in S_{1}: F(x)=i\}, i \in S_{2}$.
  \item For a set $A \subseteq V_{n}^{(p)}$, let $A^{*}=A \backslash \{0\}$ and $\chi_{a}(A)=\sum_{x \in A}\chi_{a}(x), a \in V_{n}^{(p)}$, where $\chi_{a}$ denotes the character $\chi_{a}(x)=\zeta_{p}^{\langle a, x\rangle_{n}}$.
  \item For a set $A$, let $\delta_{A}$ be the indicator function. If $A=\{a\}$, we simply denote $\delta_{\{a\}}$ by $\delta_{a}$.
\end{itemize}

\subsection{(Vectorial) bent functions}
A function from $V_{n}^{(p)}$ to $V_{m}^{(p)}$ is called a \emph{vectorial $p$-ary function}, or simply \emph{$p$-ary function} if $m=1$. For a vectorial $p$-ary function $F: V_{n}^{(p)}\rightarrow V_{m}^{(p)}$, if the sizes of $D_{F, i}, i \in V_{m}^{(p)}$, are all the same, then $F$ is called \emph{balanced}. For a $p$-ary function $f: V_{n}^{(p)}\rightarrow \mathbb{F}_{p}$, if $f(cx)=c^{l}f(x)$ for any $c \in \mathbb{F}_{p}^{*}, x \in V_{n}^{(p)}$, where $1\leq l\leq p-1$ is an integer, then $f$ is called of \emph{$l$-form}.

The \emph{Walsh transform} of a $p$-ary function $f: V_{n}^{(p)}\rightarrow \mathbb{F}_{p}$ is defined as
\begin{equation*}
W_{f}(a)=\sum_{x \in V_{n}^{(p)}}\zeta_{p}^{f(x)-\langle a, x\rangle_{n}}, a \in V_{n}^{(p)}.
\end{equation*}
A $p$-ary function $f: V_{n}^{(p)}\rightarrow \mathbb{F}_{p}$ is called a \emph{bent} function if $|W_{f}(a)|=p^{\frac{n}{2}}$ for any $a \in V_{n}^{(p)}$. For a $p$-ary bent function $f: V_{n}^{(p)}\rightarrow \mathbb{F}_{p}$, when $p=2$, then
\begin{equation*}
W_{f}(a)=2^{\frac{n}{2}}(-1)^{f^{*}(a)},
\end{equation*}
and when $p$ is odd, then
\begin{equation*}
W_{f}(a)=\left\{
\begin{split}
\pm p^{\frac{n}{2}}\zeta_{p}^{f^{*}(a)}, & \text{ if } p\equiv 1 \pmod 4 \text{ or } n \text{ is even},\\
\pm \sqrt{-1}p^{\frac{n}{2}}\zeta_{p}^{f^{*}(a)}, & \text{ if } p \equiv 3 \pmod 4 \text{ and } n \text{ is odd},
\end{split}
\right.
\end{equation*}
where $f^{*}: V_{n}^{(p)}\rightarrow \mathbb{F}_{p}$ is called the \emph{dual} of $f$. If there is a constant $\varepsilon_{f}$ such that $W_{f}(a)=\varepsilon_{f}p^{\frac{n}{2}}\zeta_{p}^{f^{*}(a)}, a \in V_{n}^{(p)}$, then the $p$-ary bent function $f$ is called \emph{weakly regular}. In particular, if $\varepsilon_{f}=1$, then $f$ is called \emph{regular}. For any weakly regular bent function $f$, the dual $f^{*}$ is also a weakly regular bent function and $\varepsilon_{f^{*}}=\varepsilon_{f}^{-1}, (f^{*})^{*}(x)=f(-x)$. Note that if $f$ is a bent function with $f(x)=f(-x)$, then $f^{*}(x)=f^{*}(-x)$.

A vectorial $p$-ary function $F: V_{n}^{(p)}\rightarrow V_{m}^{(p)}$ is called \emph{vectorial bent} if all component functions $F_{c}(x)=\langle c, F(x)\rangle_{m}, c \in {V_{m}^{(p)}}^{*}$, are $p$-ary bent functions. Note that any $p$-ary bent function $f$ is vectorial bent. A vectorial bent function $F: V_{n}^{(p)}\rightarrow V_{m}^{(p)}$ is said to be \emph{vectorial dual-bent} if there exists a vectorial bent function $G: V_{n}^{(p)}\rightarrow V_{m}^{(p)}$ such that $(F_{c})^{*}=G_{\sigma(c)}, c \in {V_{m}^{(p)}}^{*}$, where $\sigma$ is a permutation over $V_{m}^{(p)}$ with $\sigma(0)=0$. The vectorial bent function $G$ is called a \emph{vectorial dual} of $F$ and denoted by $F^{*}$. The vectorial dual of $F$ is not unique.

\subsection{Bent partitions}
Let $n$ be an even positive integer, $K$ be a positive integer divisible by $p$. Let $\Gamma=\{A_{1}, \dots, A_{K}\}$ be a partition of $V_{n}^{(p)}$. Assume that each $p$-ary function $f: V_{n}^{(p)} \rightarrow \mathbb{F}_{p}$ for which any $i \in \mathbb{F}_{p}$ has exactly $\frac{K}{p}$ of sets $A_{j}$ in $\Gamma$ in its preimage set, is a $p$-ary bent function. Then $\Gamma$ is called a \emph{bent partition of $V_{n}^{(p)}$ of depth $K$}, and every such bent function $f$ is called a \emph{bent function generated by bent partition $\Gamma$}.

In this paper, a bent partition $\Gamma$ of $V_{n}^{(p)}$ is called belonging to \emph{class $\mathcal{WBP}$} if all $p$-ary bent functions generated by $\Gamma$ are regular or all are weakly regular but not regular. We denote by $\varepsilon=\varepsilon_{f}$ for all $p$-ary bent functions $f$ generated by $\Gamma$. Note that any bent partition of $V_{n}^{(2)}$ belongs to class $\mathcal{WBP}$.

If the depth of a bent partition of $V_{n}^{(p)}$ is a power of $p$, then the following theorem holds.

\begin{theorem}\label{Theorem 1}
Let $n$ be an even positive integer, and let $\Gamma=\{A_{i}, i \in V_{m}^{(p)}\}$ be a partition of $V_{n}^{(p)}$. Define $F: V_{n}^{(p)}\rightarrow V_{m}^{(p)}$ as
\begin{equation*}
F(x)=i, \text{ if } x \in A_{i}.
\end{equation*}
Then the following statements are pairwise equivalent.

$\mathrm{(i)}$ $\Gamma$ be a bent partition (belonging to class $\mathcal{WBP}$).

$\mathrm{(ii)}$ $F$ is a vectorial bent function for which for any permutation $P$ over $V_{m}^{(p)}$, $P(F(x))$ is a vectorial bent function (with all component functions $(P(F))_{c}, c \in {V_{m}^{(p)}}^{*}$, being weakly regular and $\varepsilon_{(P(F))_{c}}=\varepsilon$, where $\varepsilon \in \{\pm 1\}$ is a constant independent of $P$ and $c$).

$\mathrm{(iii)}$ For any balanced function $B$ from $V_{m}^{(p)}$ to $\mathbb{F}_{p}$, $B(F(x))$ is a $p$-ary bent function (with $B(F(x))$ being weakly regular and $\varepsilon_{B(F(x))}=\varepsilon$, where $\varepsilon \in \{\pm 1\}$ is a constant independent of $B$).
\end{theorem}

\begin{proof}
$\mathrm{(i)}\Rightarrow \mathrm{(ii)}$: Let $P$ be any fixed permutation over $V_{m}^{(p)}$. Denote $G(x)=P(F(x))$. For any $c \in {V_{m}^{(p)}}^{*}$ and $j \in \mathbb{F}_{p}$, let $\alpha_{1}^{(c, j)}, \dots, \alpha_{p^{m-1}}^{(c, j)}$ be the solutions of the equation $\langle c, x\rangle_{m}=j$. Then for any $c \in {V_{m}^{(p)}}^{*}$, it is easy to check that
\begin{equation*}
D_{G_{c}, j}=\{A_{P^{-1}(\alpha_{1}^{(c,j)})}, A_{P^{-1}(\alpha_{2}^{(c,j)})}, \dots, A_{P^{-1}(\alpha_{p^{m-1}}^{(c,j)})}\} \text{ for any } j \in \mathbb{F}_{p}.
\end{equation*}
Thus, for any $c \in {V_{m}^{(p)}}^{*}$, $G_{c}$ is a $p$-ary bent function generated by the bent partition $\Gamma$, and then $G$ is a vectorial bent function. Furthermore, if $\Gamma$ belongs to class $\mathcal{WBP}$, then $G_{c}, c \in {V_{m}^{(p)}}^{*}$, are all weakly regular with $\varepsilon_{G_{c}}=\varepsilon$, where $\varepsilon \in \{\pm 1\}$ is a constant independent of $P$ and $c$.

$\mathrm{(ii)}\Rightarrow \mathrm{(iii)}$: Let $B$ be any fixed balanced function from $V_{m}^{(p)}$ to $\mathbb{F}_{p}$, and let $\beta \in {V_{m}^{(p)}}^{*}$ be fixed. Denote $f(x)=B(F(x))$. For any $j \in \mathbb{F}_{p}$, let $\alpha_{1}^{(j)}, \dots, \alpha_{p^{m-1}}^{(j)}$ be the solutions of the equation $\langle \beta, x\rangle_{m}=j$, and let $b_{1}^{(j)}, \dots, b_{p^{m-1}}^{(j)}$ be the solutions of the equation $B(x)=j$. Define $P: V_{m}^{(p)}\rightarrow V_{m}^{(p)}$ as $P(b_{i}^{(j)})=\alpha_{i}^{(j)}, 1\leq i \leq p^{m-1}, j \in \mathbb{F}_{p}$. Then it is easy to check that $P$ is a permutation and $B(x)=\langle \beta, P(x)\rangle_{m}$. Thus, $B(F(x))$ is the component function of the vectorial bent function $P(F(x))$, and then $B(F(x))$ is a $p$-ary bent function. Furthermore, if $(P(F))_{c}, c \in {V_{m}^{(p)}}^{*}$, are all weakly regular with $\varepsilon_{(P(F))_{c}}=\varepsilon$, where $\varepsilon \in \{\pm 1\}$ is a constant independent of $P$ and $c$, then $B(F)$ is weakly regular with $\varepsilon_{B(F)}=\varepsilon$, where $\varepsilon$ is independent of $B$.

$\mathrm{(iii)} \Rightarrow \mathrm{(i)}$: Let $f: V_{n}^{(p)}\rightarrow \mathbb{F}_{p}$ be a $p$-ary function for which for any $j \in \mathbb{F}_{p}$, $D_{f, j}=\cup_{i=1}^{p^{m-1}}A_{j_{i}}$. Define $B: V_{m}^{(p)}\rightarrow \mathbb{F}_{p}$ as $B(j_{i})=j, 1\leq i\leq p^{m-1}, j \in \mathbb{F}_{p}$. Then it is easy to check that $B$ is balanced and $f(x)=B(F(x))$. Thus, $f$ is a $p$-ary bent function, and $\Gamma$ is a bent partition. Furthermore, if for any balanced function $B: V_{m}^{(p)}\rightarrow \mathbb{F}_{p}$, $B(F)$ is weakly regular with $\varepsilon_{B(F)}=\varepsilon$, where $\varepsilon$ is a constant independent of $B$, then $\Gamma$ belongs to class $\mathcal{WBP}$.
\end{proof}

The following proposition shows that if there exists a bent partition of $V_{n}^{(p)}$ of depth $K$ (where $p \mid K$), then for any positive integer $K'$ with $p \mid K'$ and $K' \mid K$, there exists a bent partition of depth $K'$.

\begin{proposition}\label{Proposition 1}
Let $K$ and $K'$ be positive integers with $p \mid K'$ and $K' \mid K$. If there exists a bent partition of $V_{n}^{(p)}$ of depth $K$ (belonging to class $\mathcal{WBP}$), then there exists a bent partition of $V_{n}^{(p)}$ of depth $K'$ (belonging to class $\mathcal{WBP}$).
\end{proposition}

\begin{proof}
Note that $p \mid K$. Suppose that $\Gamma=\{A_{i}, 1\leq i \leq K\}$ is a bent partition of $V_{n}^{(p)}$ of depth $K$. Denote $k=\frac{K}{K'}$. Define $\Gamma'=\{B_{j}, 1\leq j\leq K'\}$, where $B_{j}=\bigcup_{i=(j-1)k+1}^{jk}A_{i}, 1\leq j \leq K'$. Let $f: V_{n}^{(p)}\rightarrow \mathbb{F}_{p}$ be a $p$-ary function for which for any $j \in \mathbb{F}_{p}$, $D_{f, j}=\bigcup_{s=1}^{\frac{K'}{p}}B_{j_{s}}$. Then $D_{f, j}=\bigcup_{s=1}^{\frac{K'}{p}}\bigcup_{i=(j_{s}-1)k+1}^{j_{s}k}A_{i}$. Since $\Gamma$ is a bent partition (belonging to class $\mathcal{WBP}$), then $f$ is a $p$-ary bent function (which is weakly regular with $\varepsilon_{f}=\varepsilon$, where $\varepsilon \in \{\pm 1\}$ is a constant independent of $f$). Therefore, $\Gamma'$ is a bent partition of depth $K'$ (belonging to class $\mathcal{WBP}$).
\end{proof}

\subsection{Generalized Hadamard matrices}
For a complex matrix $H$ of size $n \times n$ consisting of integer powers of $\zeta_{m}=e^{\frac{2\pi \sqrt{-1}}{m}}$, if $H\overline{H}^{\top}=nI_{n}$, then $H$ is called a \textit{generalized Hadamard matrix}, where $\overline{H}^{\top}$ is the transpose matrix of the conjugate matrix of $H$, and $I_{n}$ is the identity matrix of size $n \times n$. When $m=2$, $H$ is simply called a \emph{Hadamard matrix}.

For $f: V_{n}^{(p)}\rightarrow \mathbb{F}_{p}$, by \cite[Property 4]{KSW1985Ge}, $H=[\zeta_{p}^{f(x-y)}]_{x, y \in V_{n}^{(p)}}$ is a generalized Hadamard matrix if and only if $f$ is a $p$-ary bent function. Denote $U_{p}^{(1)}=\{\zeta_{p}^{i}: 0\leq i \leq p-1\}, U_{p}^{(-1)}=\{-\zeta_{p}^{i}: 0\leq i \leq p-1\}$. Let $H=[\zeta_{p}^{f(x-y)}]_{x, y \in V_{n}^{(p)}}$ be a generalized Hadamard matrix, where $f: V_{n}^{(p)}\rightarrow \mathbb{F}_{p}$ is bent and $n$ is even. With the same definition as in \cite{WFWY2024A}, the generalized Hadamard matrix $H$ is called of \emph{weakly regular type} if for all $z \in V_{n}^{(p)}$,
\begin{equation*}
p^{-\frac{n}{2}}s_{z} \in U_{p}^{(\nu_{H})},
\end{equation*}
where $s_{z}$ is the sum of the elements in a row (or column) of the generalized Hadamard matrix $H_{z}=[\zeta_{p}^{f(x-y)-\langle z, x-y\rangle_{n}}]_{x, y \in V_{n}^{(p)}}$, and $\nu_{H} \in \{\pm 1\}$ is a constant with $\nu_{H}=1$ if $p=2$.

\section{The depth of bent partitions}
In this section, we study the depth of bent partitions belonging to class $\mathcal{WBP}$. First, we give a sufficient and necessary condition for the preimage set partition of a weakly regular bent function $f: V_{n}^{(p)}\rightarrow \mathbb{F}_{p}$ to be a bent partition belonging to class $\mathcal{WBP}$, where $n$ is even and $p$ is odd.

\begin{theorem} \label{Theorem 2}
Let $n$ be an even positive integer and $p$ be an odd prime. Let $f: V_{n}^{(p)}\rightarrow \mathbb{F}_{p}$ be a weakly regular bent function. Then
$\Gamma=\{D_{f, i}, i \in \mathbb{F}_{p}\}$ is a bent partition belonging to class $\mathcal{WBP}$ if and only if
\begin{equation} \label{1}
cf^{*}(c^{-1}x)=\frac{(c+1)f^{*}(x)+(c-1)f^{*}(-x)}{2}, c \in \mathbb{F}_{p}^{*}.
\end{equation}
\end{theorem}

\begin{proof}
$\Rightarrow$:
When $p=3$, obviously Eq. (1) holds. In the following, let $p\geq 5$. By Theorem 1, for any permutation $B: \mathbb{F}_{p}\rightarrow \mathbb{F}_{p}$, $B(f(x))$ is a weakly regular bent function with $\varepsilon_{B(f)}=\varepsilon$, where $\varepsilon \in \{\pm 1\}$ is a constant. In the following, let $a \in V_{n}^{(p)}$ be fixed. We have
\begin{equation}\label{2}
\begin{split}
\varepsilon p^{\frac{n}{2}}\zeta_{p}^{(B(f))^{*}(-a)}=W_{B(f)}(-a)=\sum_{x \in V_{n}^{(p)}}\zeta_{p}^{B(f(x))+\langle a, x\rangle_{n}}=\sum_{i \in \mathbb{F}_{p}}\chi_{a}(D_{f, i})\zeta_{p}^{B(i)}.
\end{split}
\end{equation}
When $B=B_{0, j}$, where $j \in \mathbb{F}_{p}^{*}$ and $B_{0, j}=(0j)$ is the transposition permutation, by Eq. (2), we have
\begin{equation}\label{3}
\sum_{i \in \mathbb{F}_{p} \backslash \{0, j\}}\chi_{a}(D_{f, i})\zeta_{p}^{i}+\chi_{a}(D_{f, 0})\zeta_{p}^{j}+\chi_{a}(D_{f, j})=\varepsilon p^{\frac{n}{2}}\zeta_{p}^{(B_{0, j}(f))^{*}(-a)}.
\end{equation}
When $B$ is the identity map, we have
\begin{equation}\label{4}
\sum_{i \in \mathbb{F}_{p} \backslash \{0, j\}}\chi_{a}(D_{f, i})\zeta_{p}^{i}+\chi_{a}(D_{f, 0})+\chi_{a}(D_{f, j})\zeta_{p}^{j}=\varepsilon p^{\frac{n}{2}}\zeta_{p}^{f^{*}(-a)}.
\end{equation}
By Eqs. (3) and (4), for any $j \in \mathbb{F}_{p}^{*}$, we have
\begin{equation*}
(\chi_{a}(D_{f, j})-\chi_{a}(D_{f, 0}))(1-\zeta_{p}^{j})=\varepsilon p^{\frac{n}{2}}(\zeta_{p}^{(B_{0, j}(f))^{*}(-a)}-\zeta_{p}^{f^{*}(-a)}).
\end{equation*}
Denote
\begin{equation*}
A_{j}^{(a)}=\frac{1-\zeta_{p}^{(B_{0, j}(f))^{*}(-a)-f^{*}(-a)}}{1-\zeta_{p}^{j}}, j \in \mathbb{F}_{p}^{*}.
\end{equation*}
Then for any $j \in \mathbb{F}_{p}^{*}$,
\begin{equation}\label{5}
\chi_{a}(D_{f, j})=\chi_{a}(D_{f, 0})-\varepsilon p^{\frac{n}{2}}\zeta_{p}^{f^{*}(-a)}A_{j}^{(a)}.
\end{equation}
By Eqs. (2) and (5), for any permutation $B$ over $\mathbb{F}_{p}$,
\begin{equation*}
\varepsilon p^{\frac{n}{2}}\zeta_{p}^{(B(f))^{*}(-a)}=\chi_{a}(D_{f, 0})\sum_{i \in \mathbb{F}_{p}}\zeta_{p}^{B(i)}-\varepsilon p^{\frac{n}{2}}\zeta_{p}^{f^{*}(-a)}\sum_{j \in \mathbb{F}_{p}^{*}}A_{j}^{(a)}\zeta_{p}^{B(j)}=-\varepsilon p^{\frac{n}{2}}\zeta_{p}^{f^{*}(-a)}\sum_{j \in \mathbb{F}_{p}^{*}}A_{j}^{(a)}\zeta_{p}^{B(j)},
\end{equation*}
which implies that
\begin{equation} \label{6}
\sum_{j \in \mathbb{F}_{p}^{*}}A_{j}^{(a)}\zeta_{p}^{B(j)}=-\zeta_{p}^{(B(f))^{*}(-a)-f^{*}(-a)}.
\end{equation}
In particular, when $B$ is the identity map, by Eq. (6),
\begin{equation}\label{7}
\sum_{j \in \mathbb{F}_{p}^{*}}A_{j}^{(a)}\zeta_{p}^{j}=-1.
\end{equation}
If $(B_{0, j}(f))^{*}(-a)=f^{*}(-a)$, then $A_{j}^{(a)}=0$. If $(B_{0, j}(f))^{*}(-a) \neq f^{*}(-a)$, then there is $s_{j}^{(a)} \in \mathbb{F}_{p}^{*}$ such that $(B_{0, j}(f))^{*}(-a)-f^{*}(-a)=s_{j}^{(a)}j$ and $A_{j}^{(a)}=\sum_{t=0}^{s_{j}^{(a)}-1}\zeta_{p}^{tj}$. Note that $tj \in \mathbb{F}_{p}, 0\leq t\leq s_{j}^{(a)}-1$, are pairwise different and $1=-\zeta_{p}-\dots-\zeta_{p}^{p-1}$. Then there exist integers $n_{j, i}^{(a)} \in \{0, -1\}, j \in \mathbb{F}_{p}^{*}, i \in \mathbb{F}_{p}^{*}$, with
\begin{equation*}
A_{j}^{(a)}=n_{j, 1}^{(a)}\zeta_{p}+n_{j, 2}^{(a)}\zeta_{p}^{2}+\dots+n_{j, p-1}^{(a)}\zeta_{p}^{p-1}, j \in \mathbb{F}_{p}^{*}.
\end{equation*}
Since $\{\zeta_{p}, \zeta_{p}^{2}, \dots, \zeta_{p}^{p-1}\}$ is an integer basis of $\mathbb{Z}[\zeta_{p}]$, where $\mathbb{Z}[\zeta_{p}]$ is the ring of algebraic integers in the cyclotomic field $\mathbb{Q}(\zeta_{p})$, then $n_{j, i}^{(a)}, i \in \mathbb{F}_{p}^{*}$, are completely determined by $A_{j}^{(a)}$. Let $n_{0, i}^{(a)}=n_{i, 0}^{(a)}=0, i \in \mathbb{F}_{p}$. Then for any permutation $B$ over $\mathbb{F}_{p}$, by Eq. (6), we have
\begin{equation} \label{8}
\sum_{i, j \in \mathbb{F}_{p}}n_{j, i}^{(a)}\zeta_{p}^{i+B(j)}=-\zeta_{p}^{(B(f))^{*}(-a)-f^{*}(-a)}.
\end{equation}
Denote
\begin{equation*}
N_{i}^{(a, B)}=\sum_{j \in \mathbb{F}_{p}}n_{j, -B(j)+i}^{(a)}, i \in \mathbb{F}_{p},
\end{equation*}
and $t_{a, B}=(B(f))^{*}(-a)-f^{*}(-a)$. Then by Eq. (8),
\begin{equation} \label{9}
\sum_{i \in \mathbb{F}_{p} \backslash \{t_{a, B}\} }N_{i}^{(a, B)}\zeta_{p}^{i}+(N_{t_{a, B}}^{(a, B)}+1)\zeta_{p}^{t_{a, B}}=0.
\end{equation}
Since $1+x+x^2+\dots+x^{p-1}$ is the minimal polynomial of $\zeta_{p}$ over $\mathbb{Q}$, then by Eq. (9), we have
\begin{equation}\label{10}
N_{i}^{(a, B)}=N_{t_{a, B}}^{(a, B)}+1, i \in \mathbb{F}_{p} \backslash \{t_{a, B}\}.
\end{equation}
Denote $C_{a}=\sum_{j, i \in \mathbb{F}_{p}}n_{j, i}^{(a)}$. For any permutation $B$ over $\mathbb{F}_{p}$, we have
\begin{equation} \label{11}
\sum_{i \in \mathbb{F}_{p}}N_{i}^{(a, B)}=C_{a}.
\end{equation}
Then by Eqs. (10) and (11), for any permutation $B$ over $\mathbb{F}_{p}$,
\begin{equation}\label{12}
N_{i}^{(a, -B)}=\sum_{j \in \mathbb{F}_{p}}n_{j, B(j)+i}^{(a)}=\frac{C_{a}+1}{p}, i \neq t_{a, -B}, N_{t_{a, -B}}^{(a, -B)}=\sum_{j \in \mathbb{F}_{p}}n_{j, B(j)+t_{a, -B}}^{(a)}=\frac{C_{a}+1}{p}-1.
\end{equation}

If for any $j \in \mathbb{F}_{p}^{*}$, $A_{j}^{(a)}=0$ or $1$, since $1+x+x^2+\dots+x^{p-1}$ is the minimal polynomial of $\zeta_{p}$ over $\mathbb{Q}$, then by Eq. (7), we have $A_{j}^{(a)}=1, j \in \mathbb{F}_{p}^{*}$. We claim that if there is $j_{a} \in \mathbb{F}_{p}^{*}$ such that $A_{j_{a}}^{(a)} \neq 0, 1$, then $A_{j_{a}}^{(a)}=-\zeta_{p}^{-j_{a}}$ and $A_{j}^{(a)}=0, j \in \mathbb{F}_{p}^{*} \backslash \{j_{a}\}$, that is, if $n_{j_{a}, i}^{(a)}, i \in \mathbb{F}_{p}^{*}$, are not all the same, then $n_{j_{a}, -j_{a}}^{(a)}=-1$ and $n_{j, i}^{(a)}=0$ for any $i, j \in \mathbb{F}_{p}^{*}$ with $(j, i) \neq (j_{a}, -j_{a})$. In the following, let $j_{a} \in \mathbb{F}_{p}^{*}$ with $A_{j_{a}}^{(a)} \neq 0, 1$, and let $i_{a}, i_{a}' \in \mathbb{F}_{p}^{*}$ be any elements with $n_{j_{a}, i_{a}}^{(a)}=-1, n_{j_{a}, i_{a}'}^{(a)}=0$. Define
\begin{equation*}
  J_{a}=\{j \in \mathbb{F}_{p}^{*} \backslash \{j_{a}\}: A_{j}^{(a)} \neq 0\}.
\end{equation*}
\begin{itemize}
  \item Case I: $J_{a}$ is empty, that is, $n_{j, i}^{(a)}=0, j \in \mathbb{F}_{p}^{*} \backslash \{j_{a}\}, i \in \mathbb{F}_{p}^{*}$.

  In this case, $A_{j}^{(a)}=0, j \in \mathbb{F}_{p}^{*} \backslash \{j_{a}\}$. By Eq. (7), $A_{j_{a}}^{(a)}=-\zeta_{p}^{-j_{a}}$ and then $n_{j_{a}, -j_{a}}^{(a)}=-1, n_{j_{a}, i}^{(a)}=0, i \in \mathbb{F}_{p}^{*} \backslash \{-j_{a}\}$.
  \item Case II: $J_{a}$ is nonempty. We prove this case can not occur.

  To the contrary, suppose there is $j_{a}' \in J_{a}$. We first show that $(n_{j_{a}', i_{a}}^{(a)}, n_{j_{a}', i_{a}'}^{(a)})=(-1, 0)$.

  Suppose $n_{j_{a}', i_{a}}^{(a)}=n_{j_{a}', i_{a}'}^{(a)}=-1$. Let $P$ be a bijective map from $\mathbb{F}_{p} \backslash \{0, j_{a}, j_{a}'\}$ to $\mathbb{F}_{p} \backslash \{0, i_{a}, i_{a}'\}$. When $B: \mathbb{F}_{p}\rightarrow \mathbb{F}_{p}$ is the permutation defined by $B(0)=0, B(j_{a})=i_{a}, B(j_{a}')=i_{a}', B(x)=P(x), x \neq 0, j_{a}, j_{a}'$, by Eq. (12), we have
  \begin{equation*}
  \sum_{j \in \mathbb{F}_{p} \backslash \{0, j_{a}, j_{a}'\}}n_{j, P(j)}^{(a)}+n_{0, 0}^{(a)}+n_{j_{a}, i_{a}}^{(a)}+n_{j_{a}', i_{a}'}^{(a)}=\frac{C_{a}+1}{p} \ \text{or} \ \frac{C_{a}+1}{p}-1,
  \end{equation*}
  that is,
  \begin{equation} \label{13}
  \sum_{j \in \mathbb{F}_{p} \backslash \{0, j_{a}, j_{a}'\}}n_{j, P(j)}^{(a)}=\frac{C_{a}+1}{p}+2 \ \text{or} \ \frac{C_{a}+1}{p}+1.
  \end{equation}
  When  $B: \mathbb{F}_{p}\rightarrow \mathbb{F}_{p}$ is the permutation defined by $B(0)=i_{a}, B(j_{a})=i_{a}', B(j_{a}')=0, B(x)=P(x), x \neq 0, j_{a}, j_{a}'$, by Eq. (12), we have
  \begin{equation*}
  \sum_{j \in \mathbb{F}_{p} \backslash \{0, j_{a}, j_{a}'\}}n_{j, P(j)}^{(a)}+n_{0, i_{a}}^{(a)}+n_{j_{a}, i_{a}'}^{(a)}+n_{j_{a}', 0}^{(a)}=\frac{C_{a}+1}{p} \ \text{or} \ \frac{C_{a}+1}{p}-1,
  \end{equation*}
  that is,
  \begin{equation*}
  \sum_{j \in \mathbb{F}_{p} \backslash \{0, j_{a}, j_{a}'\}}n_{j, P(j)}^{(a)}=\frac{C_{a}+1}{p} \ \text{or} \ \frac{C_{a}+1}{p}-1,
  \end{equation*}
  which contradicts Eq. (13). Thus, $(n_{j_{a}', i_{a}}^{(a)}, n_{j_{a}', i_{a}'}^{(a)}) \neq (-1, -1)$.

  Suppose $n_{j_{a}', i_{a}}^{(a)}=0, n_{j_{a}', i_{a}'}^{(a)}=-1$. Let $P$ be a bijective map from $\mathbb{F}_{p} \backslash \{0, j_{a}, j_{a}'\}$ to $\mathbb{F}_{p} \backslash \{0, i_{a}, i_{a}'\}$. When $B: \mathbb{F}_{p}\rightarrow \mathbb{F}_{p}$ is the permutation with $B(0)=0, B(j_{a})=i_{a}, B(j_{a}')=i_{a}', B(x)=P(x), x \neq 0, j_{a}, j_{a}'$, by Eq. (12), we have
  \begin{equation*}
  \sum_{j \in \mathbb{F}_{p} \backslash \{0, j_{a}, j_{a}'\}}n_{j, P(j)}^{(a)}+n_{0, 0}^{(a)}+n_{j_{a}, i_{a}}^{(a)}+n_{j_{a}', i_{a}'}^{(a)}=\frac{C_{a}+1}{p} \ \text{or} \ \frac{C_{a}+1}{p}-1,
  \end{equation*}
  that is,
  \begin{equation} \label{14}
  \sum_{j \in \mathbb{F}_{p} \backslash \{0, j_{a}, j_{a}'\}}n_{j, P(j)}^{(a)}=\frac{C_{a}+1}{p}+2 \ \text{or} \ \frac{C_{a}+1}{p}+1.
  \end{equation}
  When $B: \mathbb{F}_{p}\rightarrow \mathbb{F}_{p}$ is the permutation with $B(0)=0, B(j_{a})=i_{a}', B(j_{a}')=i_{a}, B(x)=P(x), x \neq 0, j_{a}, j_{a}'$, by Eq. (12), we have
  \begin{equation*}
  \sum_{j \in \mathbb{F}_{p} \backslash \{0, j_{a}, j_{a}'\}}n_{j, P(j)}^{(a)}+n_{0, 0}^{(a)}+n_{j_{a}, i_{a}'}^{(a)}+n_{j_{a}', i_{a}}^{(a)}=\frac{C_{a}+1}{p} \ \text{or} \ \frac{C_{a}+1}{p}-1,
  \end{equation*}
  that is,
  \begin{equation*}
  \sum_{j \in \mathbb{F}_{p} \backslash \{0, j_{a}, j_{a}'\}}n_{j, P(j)}^{(a)}=\frac{C_{a}+1}{p} \ \text{or} \ \frac{C_{a}+1}{p}-1,
  \end{equation*}
  which contradicts Eq. (14). Thus, $(n_{j_{a}', i_{a}}^{(a)}, n_{j_{a}', i_{a}'}^{(a)}) \neq (0, -1)$.

  Suppose $(n_{j_{a}', i_{a}}^{(a)}, n_{j_{a}', i_{a}'}^{(a)})=(0, 0)$. Since $n_{j_{a}', i}^{(a)}, i \in \mathbb{F}_{p}^{*}$, are not all $0$, then there is $i_{a}'' \in \mathbb{F}_{p}^{*}$ with $n_{j_{a}', i_{a}''}^{(a)}=-1$. Let $P'$ be a bijective map from $\mathbb{F}_{p} \backslash \{0, j_{a}, j_{a}'\}$ to $\mathbb{F}_{p} \backslash \{i_{a}, i_{a}', i_{a}''\}$. When $B: \mathbb{F}_{p}\rightarrow \mathbb{F}_{p}$ is the permutation defined by $B(0)=i_{a}', B(j_{a})=i_{a}, B(j_{a}')=i_{a}'', B(x)=P'(x), x \neq 0, j_{a}, j_{a}'$, by Eq. (12), we have
  \begin{equation*}
  \sum_{j \in \mathbb{F}_{p} \backslash \{0, j_{a}, j_{a}'\}}n_{j, P'(j)}^{(a)}+n_{0, i_{a}'}^{(a)}+n_{j_{a}, i_{a}}^{(a)}+n_{j_{a}', i_{a}''}^{(a)}=\frac{C_{a}+1}{p} \ \text{or} \ \frac{C_{a}+1}{p}-1,
  \end{equation*}
  that is,
  \begin{equation} \label{15}
  \sum_{j \in \mathbb{F}_{p} \backslash \{0, j_{a}, j_{a}'\}}n_{j, P'(j)}^{(a)}=\frac{C_{a}+1}{p}+2 \ \text{or} \ \frac{C_{a}+1}{p}+1.
  \end{equation}
  When $B: \mathbb{F}_{p}\rightarrow \mathbb{F}_{p}$ is the permutation defined by $B(0)=i_{a}'', B(j_{a})=i_{a}', B(j_{a}')=i_{a}, B(x)=P'(x), x \neq 0, j_{a}, j_{a}'$, by Eq. (12), we have
  \begin{equation*}
  \sum_{j \in \mathbb{F}_{p} \backslash \{0, j_{a}, j_{a}'\}}n_{j, P'(j)}^{(a)}+n_{0, i_{a}''}^{(a)}+n_{j_{a}, i_{a}'}^{(a)}+n_{j_{a}', i_{a}}^{(a)}=\frac{C_{a}+1}{p} \ \text{or} \ \frac{C_{a}+1}{p}-1,
  \end{equation*}
  that is,
  \begin{equation*}
  \sum_{j \in \mathbb{F}_{p} \backslash \{0, j_{a}, j_{a}'\}}n_{j, P'(j)}^{(a)}=\frac{C_{a}+1}{p} \ \text{or} \ \frac{C_{a}+1}{p}-1,
  \end{equation*}
  which contradicts Eq. (15). Therefore, $(n_{j_{a}', i_{a}}^{(a)}, n_{j_{a}', i_{a}'}^{(a)})=(-1, 0)$. Furthermore, by the above arguments, we can obtain that for any $i \in \mathbb{F}_{p}^{*}$, $n_{j_{a}, i}^{(a)}=n_{j_{a}', i}^{(a)}$. Thus, if $j \in J_{a}$, then $A_{j}^{(a)}=A_{j_{a}}^{(a)}\neq 0, 1$. Denote $J_{a}=\{j_{a}^{(1)}, \dots, j_{a}^{(s_{a}-1)}\}$. We show that $s_{a} \neq p-1$. For any permutation $B$ over $\mathbb{F}_{p}$, we have
  \begin{equation} \label{16}
  \sum_{j \in \mathbb{F}_{p}^{*}}A_{j}^{(a)}\zeta_{p}^{B(j)}=A_{j_{a}}^{(a)}\sum_{j \in J_{a} \cup \{j_{a}\}}\zeta_{p}^{B(j)}.
  \end{equation}
  If $s_{a}=p-1$, then $J_{a} \cup \{j_{a}\}=\mathbb{F}_{p}^{*}$ and by Eq. (7), $-1=A_{j_{a}}^{(a)}\sum_{j \in \mathbb{F}_{p}^{*}}\zeta_{p}^{j}=-A_{j_{a}}^{(a)}$ and $A_{j_{a}}^{(a)}=1$, which contradicts $A_{j_{a}}^{(a)} \neq 0,1$. Hence, $2\leq s_{a}\leq p-2$. By Eqs. (6) and (16), for any permutation $B$ over $\mathbb{F}_{p}$, we obtain
  \begin{equation} \label{17}
  |\sum_{j \in J_{a} \cup \{j_{a}\}}\zeta_{p}^{B(j)}|=\frac{1}{|A_{j_{a}}^{(a)}|}.
  \end{equation}
  Denote $j_{a}^{(0)}=j_{a}$. When $B: \mathbb{F}_{p}\rightarrow \mathbb{F}_{p}$ is the permutation defined by $B(j_{a}^{(i)})=i+1, 0\leq i \leq s_{a}-1, B(x)=P''(x), x \neq j_{a}^{(i)}, 0\leq i \leq s_{a}-1$, where $P''$ is a bijective map from $\mathbb{F}_{p} \backslash \{j_{a}^{(i)}, 0\leq i \leq s_{a}-1\}$ to $\mathbb{F}_{p} \backslash \{i+1, 0\leq i \leq s_{a}-1\}$, we have
  \begin{equation*}
  \sum_{j \in J_{a} \cup \{j_{a}\}}\zeta_{p}^{B(j)}=\sum_{j=1}^{s_{a}}\zeta_{p}^{j}.
  \end{equation*}
  When $B: \mathbb{F}_{p}\rightarrow \mathbb{F}_{p}$ is the permutation defined by $B(j_{a}^{(i)})=i+1, 0\leq i \leq s_{a}-2, B(j_{a}^{(s_{a}-1)})=s_{a}+1, B(x)=P'''(x), x \neq j_{a}^{(i)}, 0\leq i \leq s_{a}-1$, where $P'''$ is a bijective map from $\mathbb{F}_{p} \backslash \{j_{a}^{(i)}, 0\leq i \leq s_{a}-1\}$ to $\mathbb{F}_{p} \backslash \{s_{a}+1, i+1, 0\leq i \leq s_{a}-2\}$, we have
   \begin{equation*}
  \sum_{j \in J_{a} \cup \{j_{a}\}}\zeta_{p}^{B(j)}=\sum_{j=1}^{s_{a}-1}\zeta_{p}^{j}+\zeta_{p}^{s_{a}+1}.
  \end{equation*}
  Denote $S_{a}=\sum_{j=1}^{s_{a}-1}\zeta_{p}^{j}$. Then by Eq. (17), we have
  \begin{equation*}
  (S_{a}+\zeta_{p}^{s_{a}})(\overline{S_{a}}+\zeta_{p}^{-s_{a}})=(S_{a}+\zeta_{p}^{s_{a}+1})(\overline{S_{a}}+\zeta_{p}^{-s_{a}-1}),
  \end{equation*}
  which implies $S_{a}=\overline{S_{a}}\zeta_{p}^{2s_{a}+1}=S_{a}\zeta_{p}^{s_{a}+1}$ and thus $s_{a}=p-1$, which is impossible. Therefore, Case II never occurs.
\end{itemize}

By the above arguments, for any fixed $a \in V_{n}^{(p)}$, $A_{j}^{(a)}=1, j \in \mathbb{F}_{p}^{*}$, or there is a unique $j_{a} \in \mathbb{F}_{p}^{*}$ such that $A_{j_{a}}^{(a)}=-\zeta_{p}^{-j_{a}}, A_{j}^{(a)}=0, j \in \mathbb{F}_{p}^{*} \backslash \{j_{a}\}$. Denote
\begin{equation*}
T=\{a \in V_{n}^{(p)}: A_{j}^{(a)}=1, j \in \mathbb{F}_{p}^{*}\}.
\end{equation*}
Define $g: V_{n}^{(p)}\rightarrow \mathbb{F}_{p}$ as
\begin{equation*}
g(-a)=\left\{
\begin{split}
0, & \ \text{if } a \in T,\\
j_{a}, & \ \text{if } a \notin T,
\end{split}
\right.
\end{equation*}
and
$h: V_{n}^{(p)}\rightarrow \mathbb{F}_{p}$ as
\begin{equation*}
h(-a)=\left\{
\begin{split}
f^{*}(-a), & \ \text{if } a \in T,\\
f^{*}(-a)-j_{a}, & \ \text{if } a \notin T.
\end{split}
\right.
\end{equation*}
When $a \in T$, for any $c \in \mathbb{F}_{p}^{*}$, by Eq. (6) and $\sum_{j \in \mathbb{F}_{p}^{*}}\zeta_{p}^{cj}=-1$, we have
\begin{equation*}
(cf)^{*}(-a)-f^{*}(-a)=0,
\end{equation*}
that is,
\begin{equation*}
(cf)^{*}(-a)=cg(-a)+h(-a).
\end{equation*}
When $a \notin T$, for any $c \in \mathbb{F}_{p}^{*}$, by Eq. (6), we have
\begin{equation*}
(cf)^{*}(-a)-f^{*}(-a)=cj_{a}-j_{a},
\end{equation*}
that is,
\begin{equation*}
(cf)^{*}(-a)=cg(-a)+h(-a).
\end{equation*}
By \cite[Theorem 1]{CM2013A}, $(cf)^{*}(x)=cf^{*}(c^{-1}x)$, and then
\begin{equation} \label{18}
cf^{*}(c^{-1}x)=cg(x)+h(x).
\end{equation}
When $c=1$, by Eq. (18), we have
\begin{equation}\label{19}
f^{*}(x)=g(x)+h(x).
\end{equation}
When $c=-1$, by Eq. (18), we have
\begin{equation} \label{20}
-f^{*}(-x)=-g(x)+h(x).
\end{equation}
Combining Eqs. (19) and (20), we obtain $g(x)=\frac{f^{*}(x)+f^{*}(-x)}{2}, h(x)=\frac{f^{*}(x)-f^{*}(-x)}{2}$, and then Eq. (1) holds.

$\Leftarrow$: Denote $g(x)=\frac{f^{*}(x)+f^{*}(-x)}{2}, h(x)=\frac{f^{*}(x)-f^{*}(-x)}{2}$. For any $a \in V_{n}^{(p)}$ and $j \in \mathbb{F}_{p}$, by \cite[Proposition 3]{WF2023Ne} and \cite[Theorem 1]{CM2013A}, we have
\begin{equation*}
\begin{split}
\chi_{a}(D_{f, j})&=p^{n-1}\delta_{0}(a)+p^{-1}\sum_{c \in \mathbb{F}_{p}^{*}}W_{cf}(-a)\zeta_{p}^{-cj}\\
&=p^{n-1}\delta_{0}(a)+\varepsilon p^{\frac{n}{2}-1}\sum_{c \in \mathbb{F}_{p}^{*}}\zeta_{p}^{(cf)^{*}(-a)-cj}\\
&=p^{n-1}\delta_{0}(a)+\varepsilon p^{\frac{n}{2}-1}\zeta_{p}^{h(-a)}\sum_{c \in \mathbb{F}_{p}^{*}}\zeta_{p}^{cg(-a)-cj}\\
&=p^{n-1}\delta_{0}(a)+\varepsilon p^{\frac{n}{2}-1}\zeta_{p}^{h(-a)}(p\delta_{g(-a)}(j)-1),
\end{split}
\end{equation*}
where $\varepsilon \in \{\pm 1\}$ is a constant. For any permutation $B$ over $\mathbb{F}_{p}$, we have
\begin{equation*}
\begin{split}
W_{B(f)}(-a)&=\sum_{j \in \mathbb{F}_{p}}\zeta_{p}^{B(j)}\chi_{a}(D_{f, j})\\
&=(p^{n-1}\delta_{0}(a)-\varepsilon p^{\frac{n}{2}-1}\zeta_{p}^{h(-a)})\sum_{j \in \mathbb{F}_{p}}\zeta_{p}^{B(j)}+\varepsilon p^{\frac{n}{2}}\zeta_{p}^{B(g(-a))+h(-a)}\\
&=\varepsilon p^{\frac{n}{2}}\zeta_{p}^{B(g(-a))+h(-a)}.
\end{split}
\end{equation*}
Then $B(f)$ is a weakly regular bent function with $\varepsilon_{B(f)}=\varepsilon$ for any permutation $B$ over $\mathbb{F}_{p}$. By Theorem 1, $\Gamma$ is a bent partition belonging to class $\mathcal{WBP}$.
\end{proof}

Based on Theorem 2, we show that the depth of a bent partition of $V_{n}^{(p)}$ belonging to class $\mathcal{WBP}$ must be a power of $p$.
\begin{theorem}\label{Theorem 3}
Let $n$ be an even positive integer and $K$ be a positive integer divided by $p$. If $\Gamma=\{A_{j}, 1\leq j \leq K\}$ is a bent partition of $V_{n}^{(p)}$ belonging to class $\mathcal{WBP}$, then $K$ is a power of $p$. In particular, the depth of any bent partition of $V_{n}^{(2)}$ must be a power of $2$.
\end{theorem}

\begin{proof} We first prove the case of $p=2$. Note that any bent partition of $V_{n}^{(2)}$ belongs to class $\mathcal{WBP}$. We only need to consider $K>4$ since the result holds if $K=2, 4$. For any Boolean bent function $f: V_{n}^{(2)}\rightarrow \mathbb{F}_{2}$ and $a \in {V_{n}^{(2)}}^{*}, j \in \mathbb{F}_{2}$, by the result in \cite{Dillon1974El}, we have
\begin{equation} \label{21}
\chi_{a}(D_{f, j})=\left\{
\begin{split}
2^{\frac{n}{2}-1}, & \text{ if } f^{*}(a)=j,\\
-2^{\frac{n}{2}-1}, & \text{ if } f^{*}(a)=j+1.
\end{split}
\right.
\end{equation}
Let $a \in {V_{n}^{(2)}}^{*}$ be fixed. Note that Eq. (21) shows that for any set $S\subseteq \{i, 1\leq i \leq K\}$ with $|S|=\frac{K}{2}$, we have
\begin{equation*}
(\chi_{a}(\cup_{i \in S}A_{i}), \chi_{a}(\cup_{i \notin S}A_{i}))=(2^{\frac{n}{2}-1}, -2^{\frac{n}{2}-1}) \text{ or } (-2^{\frac{n}{2}-1}, 2^{\frac{n}{2}-1}).
\end{equation*}
Thus, there is a function $G$ from $V_{n}^{(2)}$ to $\{1, 2, \dots, K\}$ such that
\begin{equation*}
\chi_{a}(A_{i})=\left\{
\begin{split}
-\frac{2^{\frac{n}{2}}}{K}, & \text{ if } i \in \{1, \dots, K\} \backslash \{G(a)\},\\
-\frac{2^{\frac{n}{2}}}{K}+2^{\frac{n}{2}}, & \text{ if } i=G(a),
\end{split}
\right.
\end{equation*}
or
\begin{equation*}
\chi_{a}(A_{i})=\left\{
\begin{split}
\frac{2^{\frac{n}{2}}}{K}, & \text{ if } i \in \{1, \dots, K\} \backslash \{G(a)\},\\
\frac{2^{\frac{n}{2}}}{K}-2^{\frac{n}{2}}, & \text{ if } i=G(a).
\end{split}
\right.
\end{equation*}
Since $\chi_{a}(A_{i})$ is an integer, then $K$ is a divisor of $2^{\frac{n}{2}}$ and thus $K$ is a power of $2$.

In the following, we prove the case of $p\geq 3$. Let $f$ be any bent function generated by $\Gamma$. Then there is a balanced function $B: \{1, 2, \dots, K\}\rightarrow \mathbb{F}_{p}$ such that for any $i \in \mathbb{F}_{p}$, $D_{f, i}=\cup_{j \in D_{B, i}}A_{j}$. For any permutation $P$ over $\mathbb{F}_{p}$, $P(B(x))$ is a balanced function from $\{1, 2, \dots, K\}$ to $\mathbb{F}_{p}$, and then $P(f(x))$ is a bent function generated by $\Gamma$ since $D_{P(f), i}=D_{f, P^{-1}(i)}=\cup_{j \in D_{B, P^{-1}(i)}}A_{j}=\cup_{j \in D_{P(B), i}}A_{j}$. Thus, $\{D_{f, i}, i \in \mathbb{F}_{p}\}$ is a bent partition belonging to class $\mathcal{WBP}$. For any $a \in {V_{n}^{(p)}}^{*}$ and $j \in \mathbb{F}_{p}$, by Theorem 2 and its proof, we have
\begin{equation} \label{22}
\chi_{a}(D_{f, j})=\left\{
\begin{split}
\varepsilon p^{\frac{n}{2}-1}\zeta_{p}^{h_{f}(-a)}(p-1), & \ \text{ if } g_{f}(-a)=j,\\
-\varepsilon p^{\frac{n}{2}-1}\zeta_{p}^{h_{f}(-a)}, & \ \text{ if } g_{f}(-a) \neq j,
\end{split}\right.
\end{equation}
where $\varepsilon \in \{\pm 1\}$ is a constant and $g_{f}(x)=\frac{f^{*}(x)+f^{*}(-x)}{2}, h_{f}(x)=\frac{f^{*}(x)-f^{*}(-x)}{2}$. Let $\{\alpha_{1}, \dots, \alpha_{n}\}$ be a basis of $V_{n}^{(p)}$ over $\mathbb{F}_{p}$ and define $u: \mathbb{F}_{p}^{n}\rightarrow \mathbb{F}_{p}$ as $u(x_{1}, \dots, x_{n})=f^{*}(x_{1}\alpha_{1}+\dots+x_{n}\alpha_{n})$. By Eq. (1), we have
\begin{equation} \label{23}
2cu(c^{-1}x_{1}, \dots, c^{-1}x_{n})=(c+1)u(x_{1}, \dots, x_{n})+(c-1)u(-x_{1}, \dots, -x_{n}), c \in \mathbb{F}_{p}^{*}.
\end{equation}
Assume that the algebraic normal form of $u$ is
\begin{equation*}
\begin{split}
u(x_{1}, \dots, x_{n})&=\sum_{0\leq j_{1}, \dots, j_{n}\leq p-1}a_{j_{1}, \dots, j_{n}}\prod_{i=1}^{n}x_{i}^{j_{i}}\\
&=a_{0, \dots, 0}+\sum_{d=1}^{n(p-1)}\sum_{0\leq j_{1}, \dots, j_{n}\leq p-1: j_{1}+\dots+j_{n}=d}a_{j_{1}, \dots, j_{n}}\prod_{i=1}^{n}x_{i}^{j_{i}},
\end{split}
\end{equation*}
where $a_{j_{1}, \dots, j_{n}} \in \mathbb{F}_{p}$. Then
\begin{equation} \label{24}
2cu(c^{-1}x_{1}, \dots, c^{-1}x_{n})=2ca_{0, \dots, 0}+\sum_{d=1}^{n(p-1)}\sum_{0\leq j_{1}, \dots, j_{n}\leq p-1: j_{1}+\dots+j_{n}=d}2c^{1-d}a_{j_{1}, \dots, j_{n}}\prod_{i=1}^{n}x_{i}^{j_{i}},
\end{equation}
and
\begin{equation} \label{25}
\begin{split}
&(c+1)u(x_{1}, \dots, x_{n})+(c-1)u(-x_{1}, \dots, -x_{n})\\
&=2ca_{0, \dots, 0}+\sum_{d=1}^{n(p-1)}\sum_{0\leq j_{1}, \dots, j_{n}\leq p-1: j_{1}+\dots+j_{n}=d}(c+1+(-1)^{d}(c-1))a_{j_{1}, \dots, j_{n}}\prod_{i=1}^{n}x_{i}^{j_{i}}.
\end{split}
\end{equation}
By Eqs. (23)-(25), if $a_{j_{1}, \dots, j_{n}} \neq 0$ for $(j_{1}, \dots, j_{n}) \neq (0, \dots, 0)$ with $j_{1}+\dots+j_{n}=d$, then $2c^{1-d}=c+1+(-1)^{d}(c-1)$ for all $c \in \mathbb{F}_{p}^{*}$. When $d$ is even, we have $2c^{1-d}=2c, c \in \mathbb{F}_{p}^{*}$, which implies that $d\equiv 0 \mod (p-1)$; when $d$ is odd, we have $2c^{1-d}=2, c \in \mathbb{F}_{p}^{*}$, which implies that $d\equiv 1 \mod (p-1)$. Therefore, $u$ have the following form
\begin{equation*}
u(x_{1}, \dots, x_{n})=v(x_{1}, \dots, x_{n})+w(x_{1}, \dots, x_{n}),
\end{equation*}
where
\begin{equation*}
v(x_{1}, \dots, x_{n})=a_{0, \dots, 0}+\sum_{1\leq d\leq n(p-1): d\equiv 0 \mod (p-1)}\sum_{0\leq j_{1}, \dots, j_{n}\leq p-1: j_{1}+\dots+j_{n}=d}a_{j_{1}, \dots, j_{n}}\prod_{i=1}^{n}x_{i}^{j_{i}}
\end{equation*}
and
\begin{equation*}
w(x_{1}, \dots, x_{n})=\sum_{1\leq d\leq n(p-1): d\equiv 1 \mod (p-1)}\sum_{0\leq j_{1}, \dots, j_{n}\leq p-1: j_{1}+\dots+j_{n}=d}a_{j_{1}, \dots, j_{n}}\prod_{i=1}^{n}x_{i}^{j_{i}}.
\end{equation*}
Note that $v(cx_{1}, \dots, cx_{n})=v(x_{1}, \dots, x_{n})$, $w(cx_{1}, \dots, cx_{n})=cw(x_{1}, \dots, x_{n})$ for any $c \in \mathbb{F}_{p}^{*}$. Hence, we can see that $f^{*}$ is the sum of a $(p-1)$-form and a $1$-form. Then $g_{f}(x)=\frac{f^{*}(x)+f^{*}(-x)}{2}$ is of $(p-1)$-form and $h_{f}(x)=\frac{f^{*}(x)-f^{*}(-x)}{2}$ is of $1$-form.

Case I: $-A_{i}=A_{i}, 1\leq i \leq K$. In this case, any $p$-ary bent function $f$ generated by bent partition $\Gamma$ satisfies $f(x)=f(-x)$. Thus, $f^{*}(x)=f^{*}(-x)$ and $g_{f}(x)=f^{*}(x), h_{f}(x)=0$. For any $a \in {V_{n}^{(p)}}^{*}$ and $j \in \mathbb{F}_{p}$, by Eq. (22), we have
\begin{equation*}
\chi_{a}(D_{f, j})=\left\{
\begin{split}
\varepsilon p^{\frac{n}{2}-1}(p-1), & \ \text{ if } f^{*}(a)=j,\\
-\varepsilon p^{\frac{n}{2}-1}, & \ \text{ if } f^{*}(a) \neq j.
\end{split}\right.
\end{equation*}
For any fixed $a \in {V_{n}^{(p)}}^{*}$, since $\chi_{a}(D_{f, j}), j \in \mathbb{F}_{p}$, take exactly two different values for any bent function $f$ generated by $\Gamma$, where the value $\varepsilon p^{\frac{n}{2}-1}(p-1)$ occurs one time and the value $-\varepsilon p^{\frac{n}{2}-1}$ occurs $(p-1)$ times, we obtain
\begin{equation*}
\chi_{a}(A_{i})=\left\{
\begin{split}
-\varepsilon \frac{p^{\frac{n}{2}}}{K}+\varepsilon p^{\frac{n}{2}}, & \ \text{ if } i=G(a),\\
-\varepsilon \frac{p^{\frac{n}{2}}}{K}, & \ \text{ if } i\neq G(a),
\end{split}\right.
\end{equation*}
where $G$ is some function from $V_{n}^{(p)}$ to $\{1, 2, \dots, K\}$. Since $\chi_{a}(A_{i}) \in \mathbb{Z}[\zeta_{p}] \cap \mathbb{Q}=\mathbb{Z}$, then $K$ is a divisor of $p^{\frac{n}{2}}$ and thus $K$ is a power of $p$.

Case II: There is $1\leq i_{0}\leq K$ such that $A_{i_{0}}\neq -A_{i_{0}}$. For this case, we use the proof by contradiction. Assume that $K$ is not a power of $p$.

For any $1\leq i \leq K$, $a \in {V_{n}^{(p)}}^{*}$ and $b \in \mathbb{F}_{p}^{*}$, we have
\begin{equation*}
\begin{split}
\sigma_{b}(\sum_{c \in \mathbb{F}_{p}^{*}}\chi_{ca}(A_{i}))=\sigma_{b}(\sum_{c \in \mathbb{F}_{p}^{*}}\sum_{x \in A_{i}}\zeta_{p}^{\langle ca, x\rangle_{n}})=\sum_{x \in A_{i}}\sum_{c \in \mathbb{F}_{p}^{*}}\zeta_{p}^{\langle bca, x\rangle_{n}}=\sum_{x \in A_{i}}\sum_{c \in \mathbb{F}_{p}^{*}}\zeta_{p}^{\langle ca, x\rangle_{n}}=\sum_{c \in \mathbb{F}_{p}^{*}}\chi_{ca}(A_{i}),
\end{split}
\end{equation*}
where $\sigma_{b}$ is the automorphism of the cyclotomic field $\mathbb{Q}(\zeta_{p})$ defined by $\sigma_{b}(\zeta_{p})=\zeta_{p}^{b}$. Thus, $\sum_{c \in \mathbb{F}_{p}^{*}}\chi_{ca}(A_{i}) \in \mathbb{Z}[\zeta_{p}] \cap \mathbb{Q}=\mathbb{Z}$.

Let $f$ be a bent function generated by $\Gamma$. Recall that we have proved that $g_{f}$ is of $(p-1)$-form and $h_{f}$ is of $1$-form. When $a \in {V_{n}^{(p)}}^{*}$ with $h_{f}(-a)=0$, by Eq. (22), for any $j \in \mathbb{F}_{p}$, we have
\begin{equation} \label{26}
\begin{split}
\sum_{c \in \mathbb{F}_{p}^{*}}\chi_{ca}(D_{f, j})
&=\left\{
\begin{split}
\varepsilon p^{\frac{n}{2}-1}(p-1)\sum_{c \in \mathbb{F}_{p}^{*}}\zeta_{p}^{ch_{f}(-a)}, & \ \text{ if } g_{f}(a)=j,\\
-\varepsilon p^{\frac{n}{2}-1}\sum_{c \in \mathbb{F}_{p}^{*}}\zeta_{p}^{ch_{f}(-a)}, & \ \text{ if } g_{f}(a) \neq j,
\end{split}\right.\\
&=\left\{
\begin{split}
\varepsilon p^{\frac{n}{2}-1}(p-1)^{2}, & \ \text{ if } g_{f}(a)=j,\\
-\varepsilon p^{\frac{n}{2}-1}(p-1), & \ \text{ if } g_{f}(a) \neq j.
\end{split}\right.\\
\end{split}
\end{equation}

When $a \in {V_{n}^{(p)}}^{*}$ with $h_{f}(-a) \neq 0$, by Eq. (22), for any $j \in \mathbb{F}_{p}$, we have
\begin{equation} \label{27}
\begin{split}
\sum_{c \in \mathbb{F}_{p}^{*}}\chi_{ca}(D_{f, j})
&=\left\{
\begin{split}
\varepsilon p^{\frac{n}{2}-1}(p-1)\sum_{c \in \mathbb{F}_{p}^{*}}\zeta_{p}^{ch_{f}(-a)}, & \ \text{ if } g_{f}(a)=j,\\
-\varepsilon p^{\frac{n}{2}-1}\sum_{c \in \mathbb{F}_{p}^{*}}\zeta_{p}^{ch_{f}(-a)}, & \ \text{ if } g_{f}(a) \neq j,
\end{split}\right.\\
&=\left\{
\begin{split}
-\varepsilon p^{\frac{n}{2}-1}(p-1), & \ \text{ if } g_{f}(a)=j,\\
\varepsilon p^{\frac{n}{2}-1}, & \ \text{ if } g_{f}(a) \neq j.
\end{split}\right.\\
\end{split}
\end{equation}
In particular, for any bent function $f$ generated by $\Gamma$, by Eqs. (26) and (27), we have
\begin{equation} \label{28}
\{\sum_{c \in \mathbb{F}_{p}^{*}}\chi_{ca}(D_{f, j}), j \in \mathbb{F}_{p}, a \in {V_{n}^{(p)}}^{*}\}\subseteq \{\varepsilon p^{\frac{n}{2}-1}(p-1)^{2}, -\varepsilon p^{\frac{n}{2}-1}(p-1), \varepsilon p^{\frac{n}{2}-1}\}.
\end{equation}
In the following, let $a \in {V_{n}^{(p)}}^{*}$ be fixed. Denote $S_{i}^{(a)}=\sum_{c \in \mathbb{F}_{p}^{*}}\chi_{ca}(A_{i}), 1\leq i \leq K$. Recall that we have proved that $S_{i}^{(a)}, 1\leq i \leq K$, are integers. First, we show that the following cases never occur.
\begin{itemize}
  \item $S_{i}^{(a)}, 1\leq i \leq K$, take the same value. This case contradicts Eqs. (26) and (27).
  \item $S_{i}^{(a)}, 1\leq i \leq K$, take three different values $C_{1}^{(a)}, C_{2}^{(a)}, C_{3}^{(a)}$. Let $N_{j}^{(a)} \ (1\leq j \leq 3)$ denote the frequency that $C_{j}^{(a)}$ occurs, where $N_{1}^{(a)}\geq N_{2}^{(a)}\geq N_{3}^{(a)}\geq 1$. Since $K$ is not a power of $p$ and $p\geq 3$, then $\frac{K}{p}\geq 2$ and $N_{1}^{(a)}\geq 2$. W.l.o.g, suppose $S_{1}^{(a)}=S_{2}^{(a)}=C_{1}^{(a)}, S_{3}^{(a)}=C_{2}^{(a)}, S_{4}^{(a)}=C_{3}^{(a)}$. Let $f$ be a bent function generated by $\Gamma$ for which $D_{f, 0}=A_{1}\bigcup A_{3}$ if $\frac{K}{p}=2$ and $D_{f, 0}=A_{1} \bigcup A_{3} \bigcup_{i=5}^{\frac{K}{p}+2}A_{i}$ if $\frac{K}{p}\geq 3$. Let $f'$ be a bent function generated by $\Gamma$ for which $D_{f', 0}=A_{1}\bigcup A_{4}$ if $\frac{K}{p}=2$ and $D_{f', 0}=A_{1} \bigcup A_{4} \bigcup_{i=5}^{\frac{K}{p}+2}A_{i}$ if $\frac{K}{p}\geq 3$. Let $f''$ be a bent function generated by $\Gamma$ for which $D_{f'', 0}=A_{3}\bigcup A_{4}$ if $\frac{K}{p}=2$ and $D_{f'', 0}=A_{3} \bigcup A_{4} \bigcup_{i=5}^{\frac{K}{p}+2}A_{i}$ if $\frac{K}{p}\geq 3$. Then $\sum_{c \in \mathbb{F}_{p}^{*}}\chi_{ca}(D_{f, 0}), \sum_{c \in \mathbb{F}_{p}^{*}}\chi_{ca}(D_{f', 0}), \sum_{c \in \mathbb{F}_{p}^{*}}\chi_{ca}(D_{f'', 0})$, take three different values. Let $f'''$ be a bent function generated by $\Gamma$ for which $D_{f''', 0}=A_{1}\bigcup A_{2}$ if $\frac{K}{p}=2$ and $D_{f''', 0}=A_{1} \bigcup A_{2} \bigcup_{i=5}^{\frac{K}{p}+2}A_{i}$ if $\frac{K}{p}\geq 3$. As $\sum_{c \in \mathbb{F}_{p}^{*}}\chi_{ca}(D_{f, 0}), \sum_{c \in \mathbb{F}_{p}^{*}}\chi_{ca}(D_{f', 0}), \sum_{c \in \mathbb{F}_{p}^{*}}\chi_{ca}(D_{f''', 0})$, also take three different values, by Eq. (28), we obtain $\sum_{c \in \mathbb{F}_{p}^{*}}\chi_{ca}(D_{f'', 0})=\sum_{c \in \mathbb{F}_{p}^{*}}\chi_{ca}(D_{f''', 0})$, which implies that $C_{1}^{(a)}=\frac{C_{2}^{(a)}+C_{3}^{(a)}}{2}$. If $N_{2}^{(a)}\geq 2$, then with the similar arguments as above, we obtain $C_{2}^{(a)}=\frac{C_{1}^{(a)}+C_{3}^{(a)}}{2}$ and then $C_{1}^{(a)}=C_{2}^{(a)}=C_{3}^{(a)}$, which is impossible. Thus, $N_{2}^{(a)}=N_{3}^{(a)}=1$. Let $f$ be a bent function generated by $\Gamma$ for which $A_{3}\subseteq D_{f, 0}$ and $A_{4} \subseteq D_{f, 1}$, then $\sum_{c \in \mathbb{F}_{p}^{*}}\chi_{ca}(D_{f, 0})=(\frac{K}{p}-1)C_{1}^{(a)}+C_{2}^{(a)}$, $\sum_{c \in \mathbb{F}_{p}^{*}}\chi_{ca}(D_{f, 1})=(\frac{K}{p}-1)C_{1}^{(a)}+C_{3}^{(a)}$ and $\sum_{c \in \mathbb{F}_{p}^{*}}\chi_{ca}(D_{f, j})=\frac{K}{p}C_{1}^{(a)}, j \neq 0, 1$, which contradicts Eqs. (26) and (27).
  \item $S_{i}^{(a)}, 1\leq i \leq K$, take at least four different values. This case contradicts Eq. (28).
\end{itemize}
By the above arguments, $S_{i}^{(a)}, 1\leq i \leq K$, take two different values $C_{1}^{(a)}, C_{2}^{(a)}$. Let $N_{j}^{(a)} \ (j=1, 2)$ denote the frequency that $C_{j}^{(a)}$ occurs, where $N_{1}^{(a)}\geq N_{2}^{(a)}\geq 1$. We further analyze $C_{1}^{(a)}, C_{2}^{(a)}$ by considering the following cases.
\begin{itemize}
  \item $N_{2}^{(a)}=1, N_{1}^{(a)}=K-1$. In this case, by Eqs. (26) and (27), $(\frac{K}{p}-1)C_{1}^{(a)}+C_{2}^{(a)}=\varepsilon p^{\frac{n}{2}-1}(p-1)^{2}, \frac{K}{p}C_{1}^{(a)}=-\varepsilon p^{\frac{n}{2}-1}(p-1)$ (that is, $C_{1}^{(a)}=-\frac{\varepsilon p^{\frac{n}{2}}(p-1)}{K}$, $C_{2}^{(a)}=-\frac{\varepsilon p^{\frac{n}{2}}(p-1)}{K}+\varepsilon p^{\frac{n}{2}}(p-1)$), or $(\frac{K}{p}-1)C_{1}^{(a)}+C_{2}^{(a)}=-\varepsilon p^{\frac{n}{2}-1}(p-1), \frac{K}{p}C_{1}^{(a)}=\varepsilon p^{\frac{n}{2}-1}$. If $\frac{K}{p}C_{1}^{(a)}=\varepsilon p^{\frac{n}{2}-1}$, then $C_{1}^{(a)}=\frac{\varepsilon p^{\frac{n}{2}}}{K}$ is not an integer since $K$ is not a power of $p$, which is impossible.
  \item $N_{2}^{(a)}=2, N_{1}^{(a)}=K-2$. W.l.o.g, suppose $S_{1}^{(a)}=S_{2}^{(a)}=C_{2}^{(a)}$ and $S_{i}^{(a)}=C_{1}^{(a)}, 3\leq i \leq K$. Let $f$ be a bent function generated by $\Gamma$ for which $A_{1} \subseteq D_{f, 0}, A_{2}\subseteq D_{f, 1}$. Then $\sum_{c \in \mathbb{F}_{p}^{*}}\chi_{ca}(D_{f, 0})=\sum_{c \in \mathbb{F}_{p}^{*}}\chi_{ca}(D_{f, 1})=(\frac{K}{p}-1)C_{1}^{(a)}+C_{2}^{(a)}$ and $\sum_{c \in \mathbb{F}_{p}^{*}}\chi_{ca}(D_{f, j})=\frac{K}{p}C_{1}^{(a)}, j \neq 0, 1$. When $p\geq 5$, this case contradicts Eqs. (26) and (27). When $p=3$, let $f'$ be a bent function generated by $\Gamma$ for which $A_{1} \cup A_{2} \subseteq D_{f', 0}$, then by Eqs. (26) and (27), $(\frac{K}{3}-2)C_{1}^{(a)}+2C_{2}^{(a)}=4\varepsilon 3^{\frac{n}{2}-1}, \frac{K}{3}C_{1}^{(a)}=-2\varepsilon 3^{\frac{n}{2}-1}$ (that is, $C_{1}^{(a)}=\frac{-2\varepsilon 3^{\frac{n}{2}}}{K}, C_{2}^{(a)}=\frac{-2\varepsilon 3^{\frac{n}{2}}}{K}+\varepsilon 3^{\frac{n}{2}}$), or $(\frac{K}{3}-2)C_{1}^{(a)}+2C_{2}^{(a)}=-2\varepsilon 3^{\frac{n}{2}-1},  \frac{K}{3}C_{1}^{(a)}=\varepsilon 3^{\frac{n}{2}-1}$. If $\frac{K}{3}C_{1}^{(a)}=\varepsilon 3^{\frac{n}{2}-1}$, then $C_{1}^{(a)}=\frac{\varepsilon 3^{\frac{n}{2}}}{K}$ is not an integer since $K$ is not a power of $3$, which is impossible.
  \item $N_{2}^{(a)}\geq 3, N_{1}^{(a)}=K-N_{2}^{(a)}\geq N_{2}^{(a)}$. W.l.o.g, suppose $S_{1}^{(a)}=S_{2}^{(a)}=S_{3}^{(a)}=C_{2}^{(a)}$, $S_{4}^{(a)}=S_{5}^{(a)}=S_{6}^{(a)}=C_{1}^{(a)}$. If $\frac{K}{p}=2$, then let $f$ be a bent function generated by $\Gamma$ for which $A_{1}\bigcup A_{2}=D_{f, 0}, A_{3}\bigcup A_{4}=D_{f, 1}, A_{5}\bigcup A_{6}=D_{f, 2}$. In this case, $\sum_{c \in \mathbb{F}_{p}^{*}}\chi_{ca}(D_{f, j}), j=0, 1, 2$, are different, which contradicts Eqs. (26) and (27). If $\frac{K}{p}\geq 3$, then let $f$ be a bent function generated by $\Gamma$ for which $D_{f, 0}=A_{1}\bigcup A_{2}\bigcup A_{3}$ if $\frac{K}{p}=3$ and $D_{f, 0}=A_{1} \bigcup A_{2} \bigcup A_{3} \bigcup_{i=7}^{\frac{K}{p}+3}A_{i}$ if $\frac{K}{p}\geq 4$; let $f'$ be a bent function generated by $\Gamma$ for which $D_{f', 0}=A_{1}\bigcup A_{2}\bigcup A_{4}$ if $\frac{K}{p}=3$ and $D_{f', 0}=A_{1} \bigcup A_{2} \bigcup A_{4} \bigcup_{i=7}^{\frac{K}{p}+3}A_{i}$ if $\frac{K}{p}\geq 4$; let $f''$ be a bent function generated by $\Gamma$ for which $D_{f'', 0}=A_{1}\bigcup A_{4}\bigcup A_{5}$ if $\frac{K}{p}=3$ and $D_{f'', 0}=A_{1} \bigcup A_{4} \bigcup A_{5} \bigcup_{i=7}^{\frac{K}{p}+3}A_{i}$ if $\frac{K}{p}\geq 4$; and let $f'''$ be a bent function generated by $\Gamma$ for which $D_{f''', 0}=A_{4}\bigcup A_{5}\bigcup A_{6}$ if $\frac{K}{p}=3$ and $D_{f''', 0}=A_{4} \bigcup A_{5} \bigcup A_{6}\bigcup_{i=7}^{\frac{K}{p}+3}A_{i}$ if $\frac{K}{p}\geq 4$. Then $\sum_{c \in \mathbb{F}_{p}^{*}}\chi_{ca}(D_{f, 0}), \sum_{c \in \mathbb{F}_{p}^{*}}\chi_{ca}(D_{f', 0}), \sum_{c \in \mathbb{F}_{p}^{*}}\chi_{ca}(D_{f'', 0}), \sum_{c \in \mathbb{F}_{p}^{*}}\chi_{ca}(D_{f''', 0})$, are different, which contradicts Eq. (28).
\end{itemize}
By the above arguments, when $p\geq 5$, only Eq. (26) occurs and $h_{f}(a)=0$ for all $a \in {V_{n}^{(p)}}^{*}$, where $f$ is any bent function generated by $\Gamma$. Since $h_{f}(x)=\frac{f^{*}(x)-f^{*}(-x)}{2}$, then $f^{*}(x)=f^{*}(-x), x \in V_{n}^{(p)}$, which implies that $f(x)=f(-x), x \in V_{n}^{(p)}$. Recall that $A_{i_{0}}\neq -A_{i_{0}}$. Then there is nonzero $a_{0} \in A_{i_{0}}$ with $-a_{0} \notin A_{i_{0}}$. Suppose $-a_{0} \in A_{i_{1}}$. Let $f$ be a bent function generated by $\Gamma$ for which $A_{i_{0}}\subseteq D_{f, 0}$ and $A_{i_{1}} \subseteq D_{f, 1}$. Then $f(a_{0})=0, f(-a_{0})=1$, which is a contradiction. Therefore, $K$ must be a power of $p$ when $p\geq 5$.

When $p=3$, for any fixed $a \in {V_{n}^{(3)}}^{*}$, we have obtained $(N_{1}^{(a)}, N_{2}^{(a)}, C_{1}^{(a)}, C_{2}^{(a)})=(K-1, 1, -\frac{2\varepsilon 3^{\frac{n}{2}}}{K}, -\frac{2\varepsilon 3^{\frac{n}{2}}}{K}+2 \varepsilon 3^{\frac{n}{2}})$ or $(N_{1}^{(a)}, N_{2}^{(a)}, C_{1}^{(a)}, C_{2}^{(a)})=(K-2, 2, -\frac{2\varepsilon 3^{\frac{n}{2}}}{K}, -\frac{2\varepsilon 3^{\frac{n}{2}}}{K}+\varepsilon 3^{\frac{n}{2}})$. Since $K$ is not a power of $3$ and $\frac{2\varepsilon 3^{\frac{n}{2}}}{K}$ is an integer, then $K=2 \cdot 3^m$ for some positive integer $m\leq \frac{n}{2}$. By Proposition 1, there exists a bent partition $\Gamma'$ of depth $6$ belonging to class $\mathcal{WBP}$. For the consistency of notation, we also denote $\Gamma'=\{A_{i}, 1\leq i \leq 6\}$ (that is equivalent to considering the case of $K=6$). If $(N_{1}^{(a)}, N_{2}^{(a)}, C_{1}^{(a)}, C_{2}^{(a)})=(5, 1, -\varepsilon 3^{\frac{n}{2}-1}, 5\varepsilon  3^{\frac{n}{2}-1})$ for all $a \in {V_{n}^{(3)}}^{*}$, then only Eq. (26) holds, which is impossible with the similar proof as in the case of $p\geq 5$. Thus, there exists $a_{0} \in {V_{n}^{(3)}}^{*}$ such that $(N_{1}^{(a_{0})}, N_{2}^{(a_{0})}, C_{1}^{(a_{0})}, C_{2}^{(a_{0})})=(4, 2, -\varepsilon 3^{\frac{n}{2}-1}, 2\varepsilon 3^{\frac{n}{2}-1})$. W.l.o.g., suppose $S_{1}^{(a_{0})}=S_{2}^{(a_{0})}=2\varepsilon 3^{\frac{n}{2}-1}$. Note that $S_{i}^{(a_{0})}=\chi_{a_{0}}(A_{i})+\overline{\chi_{a_{0}}(A_{i})}, 1\leq i \leq 6$ as $p=3$. Then $\chi_{a_{0}}(A_{i})=\varepsilon 3^{\frac{n}{2}-1}+\beta_{i}\sqrt{-1}$ for $i=1, 2$, and $\chi_{a_{0}}(A_{i})=-\frac{\varepsilon 3^{\frac{n}{2}-1}}{2}+\beta_{i}\sqrt{-1}$ for $3\leq i\leq 6$, where $\beta_{i}, 1\leq i \leq 6$, are real numbers. For a complex number $a+b \sqrt{-1}$, where $a, b$ are real numbers, denote $Re(a+b \sqrt{-1})=a, Im(a+b \sqrt{-1})=b$. For any $3\leq i \leq 5$, let $f^{(i)}$ be the bent function generated by $\Gamma'$ for which $D_{f^{(i)}, 0}=A_{1}\cup A_{i}$, $D_{f^{(i)}, 1}=A_{2} \cup A_{6}$. By Eq. (22) and $Re(\chi_{a_{0}}(D_{f^{(i)}, 0}))=Re(\chi_{a_{0}}(D_{f^{(i)}, 1}))$, we have $\chi_{a_{0}}(D_{f^{(i)}, 0})=\chi_{a_{0}}(D_{f^{(i)}, 1})$, and then $\beta_{1}+\beta_{i}=\beta_{2}+\beta_{6}$ for any $3\leq i \leq 5$, which implies that $\beta_{3}=\beta_{4}=\beta_{5}$. Let $f$ be the bent function generated by $\Gamma'$ for which $D_{f, 0}=A_{1}\cup A_{6}$, $D_{f, 1}=A_{2} \cup A_{3}$. By Eq. (22) and $Re(\chi_{a_{0}}(D_{f,0}))=Re(\chi_{a_{0}}(D_{f, 1}))$, we have $\chi_{a_{0}}(D_{f, 0})=\chi_{a_{0}}(D_{f, 1})$, and then $\beta_{1}+\beta_{6}=\beta_{2}+\beta_{3}$. Combining $\beta_{1}+\beta_{6}=\beta_{2}+\beta_{3}$, $\beta_{1}+\beta_{3}=\beta_{2}+\beta_{6}$ and
$\beta_{3}=\beta_{4}=\beta_{5}$, we obtain $\beta_{1}=\beta_{2}, \beta_{3}=\beta_{4}=\beta_{5}=\beta_{6}$. Let $f'$ be a bent function generated by $\Gamma'$ with $D_{f', 0}=A_{1}\cup A_{2}$. Then by Eq. (22) and $Re(\chi_{a_{0}}(D_{f', 0}))=2\varepsilon 3^{\frac{n}{2}-1}$, we have $h_{f'}(-a_{0})=0$, and then $\chi_{a_{0}}(D_{f', 1})=-\varepsilon 3^{\frac{n}{2}-1}+2\beta_{3} \sqrt{-1}$ is an integer, which implies that $\beta_{3}=0$. Then $\chi_{a_{0}}(A_{3})=-\frac{\varepsilon 3^{\frac{n}{2}-1}}{2} \in \mathbb{Z}[\zeta_{3}]\cap \mathbb{Q}=\mathbb{Z}$, which is impossible. Hence, there is no bent partition of depth $6$ belonging to class $\mathcal{WBP}$, and the depth of any bent partition of $V_{n}^{(3)}$ belonging to class $\mathcal{WBP}$ is a power of $3$. This completes the proof.
\end{proof}

\begin{remark}\label{Remark 1}
By Theorems 1 and 2, bent partitions $\Gamma$ of $V_{n}^{(p)}$ of depth $K$ belonging to class $\mathcal{WBP}$ one-to-one correspond to vectorial bent functions $F: V_{n}^{(p)}\rightarrow V_{m}^{(p)}$ ($m=log_{p}K$) for which one of the following conditions holds:
\begin{itemize}
  \item $P(F(x))$ is a vectorial bent function whose component functions are all weakly regular with $\varepsilon_{(P(F))_{c}}=\varepsilon$ (where $\varepsilon \in \{\pm 1\}$ is a constant independent of $P$ and $c$) for any permutation $P$ over $V_{m}^{(p)}$.
  \item $B(F(x))$ is a weakly regular $p$-ary bent function with $\varepsilon_{B(F(x))}=\varepsilon$ (where $\varepsilon \in \{\pm 1\}$ is a constant independent of $B$) for any balanced function $B: V_{m}^{(p)}\rightarrow \mathbb{F}_{p}$.
\end{itemize}
In particular, bent partitions $\Gamma$ of $V_{n}^{(2)}$ of depth $K$ one-to-one correspond to Boolean vectorial bent functions $F: V_{n}^{(2)}\rightarrow V_{m}^{(2)}$ ($m=log_{2}K$) for which $P(F(x))$ is a Boolean vectorial bent function for any permutation $P$ over $V_{m}^{(2)}$, or $B(F(x))$ is a Boolean bent function for any balanced function $B: V_{m}^{(2)}\rightarrow \mathbb{F}_{2}$.
\end{remark}

\section{New constructions of bent partitions}
In this section, we will provide new constructions of bent partitions. In particular, for any characteristic $p$, bent partitions that do not correspond to vectorial dual-bent functions will be obtained.

\begin{theorem} \label{Theorem 4}
Let $n, m$ be positive integers with $2 \mid n, m\leq \frac{n}{2}$, and $m\geq 2$ if $p=2$. Let $F: V_{n}^{(p)}\rightarrow V_{m}^{(p)}$ be a vectorial bent function for which any component function $F_{c}$ ($c \in {V_{m}^{(p)}}^{*}$) is weakly regular with $\varepsilon_{F_{c}}=\varepsilon$, where $\varepsilon \in \{\pm 1\}$ is a constant independent of $c$, and there are functions $G: V_{n}^{(p)}\rightarrow V_{m}^{(p)}$ and $h: V_{n}^{(p)}\rightarrow \mathbb{F}_{p}$ such that
\begin{equation} \label{29}
(F_{c})^{*}(x)=G_{c}(x)+h(x), c \in {V_{m}^{(p)}}^{*}, x \in V_{n}^{(p)}.
\end{equation}
Then $\{D_{F, i}, i \in V_{m}^{(p)}\}$ is a bent partition belonging to class $\mathcal{WBP}$.
\end{theorem}

\begin{proof}
For any $a \in V_{n}^{(p)}$ and $i \in V_{m}^{(p)}$, by \cite[Proposition 3]{WF2023Ne}, we have
\begin{equation} \label{30}
\begin{split}
\chi_{a}(D_{F, i})&=p^{n-m}\delta_{0}(a)+p^{-m}\sum_{c \in {V_{m}^{(p)}}^{*}}W_{F_{c}}(-a)\zeta_{p}^{-\langle c, i\rangle_{m}}\\
&=p^{n-m}\delta_{0}(a)+\varepsilon p^{\frac{n}{2}-m}\sum_{c \in {V_{m}^{(p)}}^{*}}\zeta_{p}^{(F_{c})^{*}(-a)-\langle c, i\rangle_{m}}\\
&=p^{n-m}\delta_{0}(a)+\varepsilon p^{\frac{n}{2}-m}\zeta_{p}^{h(-a)}\sum_{c \in {V_{m}^{(p)}}^{*}}\zeta_{p}^{\langle c, G(-a)-i\rangle_{m}}\\
&=p^{n-m}\delta_{0}(a)+\varepsilon p^{\frac{n}{2}-m}\zeta_{p}^{h(-a)}(p^m\delta_{G(-a)}(i)-1).
\end{split}
\end{equation}
For any union $D$ of $p^{m-1}$ sets of $\{D_{F, i}, i \in V_{m}^{(p)}\}$, we have
\begin{equation*}
\chi_{a}(D)=\left\{
\begin{split}
p^{n-1}\delta_{0}(a)+\varepsilon p^{\frac{n}{2}-1}\zeta_{p}^{h(-a)}(p-1), & \text{ if } D_{F, G(-a)} \subseteq D,\\
p^{n-1}\delta_{0}(a)-\varepsilon p^{\frac{n}{2}-1}\zeta_{p}^{h(-a)}, & \text{ if } D_{F, G(-a)} \nsubseteq D.
\end{split}
\right.
\end{equation*}
Let $f: V_{n}^{(p)} \rightarrow \mathbb{F}_{p}$ be a function for which for each $j \in \mathbb{F}_{p}$, there are exactly $p^{m-1}$ sets $D_{F, i}$ in its preimage set. For any $a \in V_{n}^{(p)}$, we have
\begin{equation*}
\chi_{a}(D_{f, j})=\left\{
\begin{split}
p^{n-1}\delta_{0}(a)+\varepsilon p^{\frac{n}{2}-1}\zeta_{p}^{h(-a)}(p-1), & \text{ if } j=g(a),\\
p^{n-1}\delta_{0}(a)-\varepsilon p^{\frac{n}{2}-1}\zeta_{p}^{h(-a)}, & \text{ if } j \neq g(a),
\end{split}
\right.
\end{equation*}
where $g(a)=f(D_{F, G(-a)})$. Then we can obtain
\begin{equation*}
W_{f}(-a)=\sum_{j \in \mathbb{F}_{p}}\zeta_{p}^{j}\chi_{a}(D_{f, j})=\varepsilon p^{\frac{n}{2}}\zeta_{p}^{h(-a)+g(a)}.
\end{equation*}
Thus $f$ is weakly regular bent with $\varepsilon_{f}=\varepsilon$ and $\{D_{F, i}, i \in V_{m}^{(p)}\}$ is a bent partition belonging to class $\mathcal{WBP}$.
\end{proof}

\begin{remark} \label{Remark 2}
When $p=2$, by \cite[Proposition 2]{WFWY2024A}, Eq. (29) is a necessary and sufficient condition for the preimage set partition of a vectorial bent function to be a bent partition. When $p$ is odd and $m=1$, by Theorem 2, Eq. (29) is a necessary and sufficient condition for the preimage set partition of a weakly regular bent function to be a bent partition belonging to class $\mathcal{WBP}$. For other cases, whether Eq. (29) is a necessary and sufficient condition for the preimage set partition of a vectorial bent function whose component functions are all regular or all weakly regular but not regular to be a bent partition belonging to class $\mathcal{WBP}$ is open.
\end{remark}

\begin{remark} \label{Remark 3}
Note that the vectorial bent function $F$ in Theorem 4 is vectorial dual-bent if and only if $h$ is the zero function.
\end{remark}

We characterize the conditions in Theorem 4 in terms of generalized Hadamard matrices.

\begin{proposition} \label{Proposition 2}
Let $F: V_{n}^{(p)}\rightarrow V_{m}^{(p)}$, where $n$ is even, $m\leq \frac{n}{2}$, and $m\geq 2$ if $p=2$. Define
\begin{equation*}
H_{c}=\left[\zeta_{p}^{F_{c}(x-y)}\right]_{x, y \in V_{n}^{(p)}}, c \in {V_{m}^{(p)}}^{*}.
\end{equation*}
The following two statements are equivalent.

$\mathrm{(i)}$ $F$ is a vectorial bent function for which all component functions $F_{c}, c \in {V_{m}^{(p)}}^{*}$, are weakly regular with $\varepsilon_{F_{c}}=\varepsilon$ and there are functions $G: V_{n}^{(p)} \rightarrow V_{m}^{(p)}$ and $h: V_{n}^{(p)} \rightarrow \mathbb{F}_{p}$ such that $(F_{c})^{*}(x)=G_{c}(x)+h(x), c \in {V_{m}^{(p)}}^{*}, x \in V_{n}^{(p)}$, where $\varepsilon \in \{\pm 1\}$ is a constant independent of $c$.

$\mathrm{(ii)}$ $H_{c}, c \in {V_{m}^{(p)}}^{*}$, are generalized Hadamard matrices of weakly regular type with $\nu_{H_{c}}=\varepsilon$ and $H_{d}\overline{H_{e}}^{\top} \overline{H_{d-e}}^{\top}, d \neq e \in {V_{m}^{(p)}}^{*}$, are all the same, where $\varepsilon \in \{\pm 1\}$ is a constant independent of $c$.
\end{proposition}

\begin{proof}
$\mathrm{(i)} \Rightarrow \mathrm{(ii)}$: Since $F_{c}$ is weakly regular bent with $\varepsilon_{F_{c}}=\varepsilon$, then by
\cite[Property 4]{KSW1985Ge}, $H_{c}$ is a generalized Hadamard matrix of weakly regular type with $\nu_{H_{c}}=\varepsilon$. For a matrix $M=[a_{i, j}]$, denote $a_{i, j}$ by $(M)_{i, j}$. Since $F$ is vectorial bent with $\varepsilon_{F_{c}}=\varepsilon, c \in {V_{m}^{(p)}}^{*}$, and $(F_{c})^{*}(x)=G_{c}(x)+h(x)$, then for any $d \neq e \in {V_{m}^{(p)}}^{*}$ and $i, j \in V_{n}^{(p)}$, we have
\begin{equation*} \label{31}
\begin{split}
&(H_{d}\overline{H_{e}}^{\top} \overline{H_{d-e}}^{\top})_{i, j}\\
&=\sum_{u \in V_{n}^{(p)}}(H_{d}\overline{H_{e}}^{\top})_{i, u}(\overline{H_{d-e}}^{\top})_{u, j}\\
&=\sum_{u \in V_{n}^{(p)}}\sum_{v \in V_{n}^{(p)}}(H_{d})_{i, v}(\overline{H_{e}}^{\top})_{v, u}(\overline{H_{d-e}}^{\top})_{u, j}\\
&=\sum_{u \in V_{n}^{(p)}}\sum_{v \in V_{n}^{(p)}}(H_{d})_{i, v}(\overline{H_{e}})_{u, v}(\overline{H_{d-e}})_{j, u}\\
&=\sum_{u \in V_{n}^{(p)}}\sum_{v \in V_{n}^{(p)}}\zeta_{p}^{F_{d}(i-v)-F_{e}(u-v)-F_{d-e}(j-u)}\\
&=\varepsilon p^{-\frac{3n}{2}}\sum_{u \in V_{n}^{(p)}}\sum_{v \in V_{n}^{(p)}}\sum_{x \in V_{n}^{(p)}}\zeta_{p}^{(F_{d})^{*}(x)+\langle i-v, x\rangle_{n}}\sum_{y \in V_{n}^{(p)}}\zeta_{p}^{-(F_{e})^{*}(y)-\langle u-v, y\rangle_{n}}\sum_{z \in V_{n}^{(p)}}\zeta_{p}^{-(F_{d-e})^{*}(z)-\langle j-u, z\rangle_{n}}\\
&=\varepsilon p^{-\frac{3n}{2}}\sum_{x, y, z \in V_{n}^{(p)}}\zeta_{p}^{(F_{d})^{*}(x)-(F_{e})^{*}(y)-(F_{d-e})^{*}(z)+\langle i, x\rangle_{n}-\langle j, z\rangle_{n}}\sum_{u \in V_{n}^{(p)}}\zeta_{p}^{\langle u, z-y\rangle_{n}}\sum_{v \in V_{n}^{(p)}}\zeta_{p}^{\langle v, y-x\rangle_{n}}\\
\end{split}
\end{equation*}
\begin{equation}
\begin{split}
&=\varepsilon p^{\frac{n}{2}}\sum_{x \in V_{n}^{(p)}}\zeta_{p}^{(F_{d})^{*}(x)-(F_{e})^{*}(x)-(F_{d-e})^{*}(x)+\langle i-j, x\rangle_{n}}\\
&=\varepsilon p^{\frac{n}{2}}\sum_{x \in V_{n}^{(p)}}\zeta_{p}^{-h(x)+\langle i-j, x\rangle_{n}},
\end{split}
\end{equation}
where in the fifth equation we use the fact that if $f$ is weakly regular bent, then $f^{*}$ is weakly regular bent with $\varepsilon_{f^{*}}=\varepsilon_{f}^{-1}$ and $(f^{*})^{*}(x)=f(-x)$. By Eq. (31), statement $\mathrm{(ii)}$ holds.

$\mathrm{(ii)} \Rightarrow \mathrm{(i)}$: Since $H_{c}, c \in {V_{m}^{(p)}}^{*}$, are generalized Hadamard matrices of weakly regular type with $\nu_{H_{c}}=\varepsilon$, then by \cite[Property 4]{KSW1985Ge}, $F$ is vectorial bent for which all component functions $F_{c}, c \in {V_{m}^{(p)}}^{*}$, are weakly regular with $\varepsilon_{F_{c}}=\varepsilon$. For any $i, j \in V_{n}^{(p)}$ and $d \neq e \in {V_{m}^{(p)}}^{*}$, with the same computation as in Eq. (31),
\begin{equation*}
(H_{d}\overline{H_{e}}^{\top} \overline{H_{d-e}}^{\top})_{i, j}=\varepsilon p^{\frac{n}{2}}\sum_{x \in V_{n}^{(p)}}\zeta_{p}^{(F_{d})^{*}(x)-(F_{e})^{*}(x)-(F_{d-e})^{*}(x)+\langle i-j, x\rangle_{n}}.
\end{equation*}
Let $\alpha \in V_{m}^{(p)} \backslash \mathbb{F}_{p}$, and denote $h(x)=-(F_{1})^{*}(x)+(F_{\alpha})^{*}(x)+(F_{1-\alpha})^{*}(x)$. Since $H_{d}\overline{H_{e}}^{\top} \overline{H_{d-e}}^{\top}, d\neq e \in {V_{m}^{(p)}}^{*}$, are all the same, then for any $d \neq e \in {V_{m}^{(p)}}^{*}$,
\begin{equation*}
\sum_{x \in V_{n}^{(p)}}\zeta_{p}^{(F_{d})^{*}(x)-(F_{e})^{*}(x)-(F_{d-e})^{*}(x)+\langle i, x\rangle_{n}}=\sum_{x \in V_{n}^{(p)}}\zeta_{p}^{-h(x)+\langle i, x\rangle_{n}}, i \in V_{n}^{(p)},
\end{equation*}
which implies that
\begin{equation} \label{32}
(F_{d})^{*}(x)-(F_{e})^{*}(x)-(F_{d-e})^{*}(x)=-h(x), d \neq e \in {V_{m}^{(p)}}^{*}.
\end{equation}
Let $\{\alpha_{1}, \dots, \alpha_{m}\}$ be a basis of $V_{m}^{(p)}$ over $\mathbb{F}_{p}$. For any $x \in V_{n}^{(p)}$, let $G(x) \in V_{m}^{(p)}$ be defined by the following system:
\begin{equation*}
\left\{
\begin{split}
& \langle \alpha_{1}, G(x)\rangle_{m}=(F_{\alpha_{1}})^{*}(x)-h(x),\\
& \langle \alpha_{2}, G(x)\rangle_{m}=(F_{\alpha_{2}})^{*}(x)-h(x),\\
& \ \ \ \ \ \ \vdots\\
& \langle \alpha_{m}, G(x)\rangle_{m}=(F_{\alpha_{m}})^{*}(x)-h(x).
\end{split}
\right.
\end{equation*}
For any $c \in {V_{m}^{(p)}}^{*}$, denote $c=a_{i_{1}}\alpha_{i_{1}}+\dots+a_{i_{t}}\alpha_{i_{t}}$, where $a_{i_{j}} \in \mathbb{F}_{p}^{*}, 1\leq j \leq t, 1\leq i_{1} < \dots <i_{t}\leq m, 1\leq t\leq m$. Then
\begin{equation} \label{33}
G_{c}(x)=\sum_{j=1}^{t}a_{i_{j}}\langle \alpha_{i_{j}}, G(x)\rangle_{m}=\sum_{j=1}^{t}a_{i_{j}}(F_{\alpha_{i_{j}}})^{*}(x)-h(x)\sum_{j=1}^{t}a_{i_{j}}.
\end{equation}
For any $a \in \mathbb{F}_{p}^{*}$ and $c \in {V_{m}^{(p)}}^{*}$, we first show
\begin{equation} \label{34}
(F_{ac})^{*}(x)=a(F_{c})^{*}(x)-(a-1)h(x).
\end{equation}
Obviously, when $a=1$, Eq. (34) holds. Suppose Eq. (34) holds for $a=k$, where $1\leq k \leq p-2$. When $a=k+1$, by Eq. (32), we have
\begin{equation*}
\begin{split}
(F_{(k+1)c})^{*}(x)&=(F_{kc})^{*}(x)+(F_{c})^{*}(x)-h(x)\\
&=k(F_{c})^{*}(x)-(k-1)h(x)+(F_{c})^{*}(x)-h(x)\\
&=(k+1)(F_{c})^{*}(x)-kh(x).
\end{split}
\end{equation*}
Then Eq. (34) also holds for $a=k+1$. By induction, Eq. (34) holds. When $t=1$, by Eqs. (33) and (34),
\begin{equation} \label{35}
(F_{c})^{*}(x)=G_{c}(x)+h(x).
\end{equation}
Suppose Eq. (35) holds for $t=s$. When $t=s+1$, denote $c'=\sum_{j=1}^{s}a_{i_{j}}\alpha_{i_{j}}$, by Eq. (32), we have
\begin{equation*}
\begin{split}
(F_{c})^{*}(x)&=(F_{c'+a_{i_{s+1}}\alpha_{i_{s+1}}})^{*}(x)\\
&=(F_{c'})^{*}(x)+(F_{a_{i_{s+1}}\alpha_{i_{s+1}}})^{*}(x)-h(x)\\
&=G_{c'}(x)+h(x)+G_{a_{i_{s+1}}\alpha_{i_{s+1}}}(x)+h(x)-h(x)\\
&=G_{c}(x)+h(x).
\end{split}
\end{equation*}
Then Eq. (35) also holds for $t=s+1$. By induction, Eq. (35) holds.
\end{proof}

For bent partitions of $V_{n}^{(2)}$ corresponding to vectorial dual-bent functions, a characterization in terms of Hadamard matrices was given in \cite{WFWY2024A}. For general bent partitions of $V_{n}^{(2)}$, we give a characterization in terms of Hadamard matrices, which directly follows from \cite[Proposition 2]{WFWY2024A} and Proposition 2.

\begin{theorem} \label{Theorem 5}
Let $F: V_{n}^{(2)}\rightarrow V_{m}^{(2)}$, where $n$ is even and $2\leq m\leq \frac{n}{2}$. Define
\begin{equation*}
H_{c}=\left[\zeta_{p}^{F_{c}(x+y)}\right]_{x, y \in V_{n}^{(2)}}, c \in {V_{m}^{(2)}}^{*}.
\end{equation*}
Then $\{D_{F, i}, i \in V_{m}^{(2)}\}$ is a bent partition if and only if $H_{c}, c \in {V_{m}^{(2)}}^{*}$, are Hadamard matrices for which $H_{d}H_{e}H_{d+e}, d \neq e \in {V_{m}^{(2)}}^{*}$, are all the same.
\end{theorem}

By Theorem 4, the proof of Proposition 2, and Remark 3, we directly obtain the following result.

\begin{corollary} \label{Corollary 1}
Let $n, m$ be positive integers with $2 \mid n, m\leq \frac{n}{2}$, and $m\geq 2$ if $p=2$. Let $F: V_{n}^{(p)}\rightarrow V_{m}^{(p)}$ be a vectorial bent function for which all component functions are regular or all are weakly regular but not regular. If the $p$-ary functions $(F_{d})^{*}(x)+(F_{e})^{*}(x)-(F_{d+e})^{*}(x), d\neq -e \in {V_{m}^{(p)}}^{*}$, are all the same, then $\{D_{F, i}, i \in V_{m}^{(p)}\}$ is a bent partition belonging to class $\mathcal{WBP}$. Furthermore, if $(F_{d})^{*}(x)+(F_{e})^{*}(x)-(F_{d+e})^{*}(x), d\neq -e \in {V_{m}^{(p)}}^{*}$, are the same nonzero function, then $\{D_{F, i}, i \in V_{m}^{(p)}\}$ is a bent partition that do not correspond to vectorial dual-bent functions.
\end{corollary}

In the following, based on Theorem 4 and Corollary 1, new constructions of bent partitions are provided.

In \cite{AFKMWW2025A}, an explicit construction of bent partitions not corresponding to vectorial dual-bent functions was given. Based on Corollary 1, we generalize the construction.

\begin{proposition} \label{Proposition 3}
Let $m, n$ be positive integers with $m \mid n, m \neq n$, and $m\geq 2$ if $p=2$. Let $\pi$ be a permutation over $\mathbb{F}_{p^n}$ with $\pi(cx)=\theta(c)\pi(x), c \in \mathbb{F}_{p^m}^{*}, x \in \mathbb{F}_{p^n}$, where $\theta$ is a permutation over $\mathbb{F}_{p^m}^{*}$ with $\theta^{-1}(d^{-1})+\theta^{-1}(e^{-1})=\theta^{-1}((d+e)^{-1})$ for any $d\neq -e \in \mathbb{F}_{p^m}^{*}$. Define $F: \mathbb{F}_{p^n} \times \mathbb{F}_{p^n}\rightarrow \mathbb{F}_{p^m}$ as
\begin{equation*}
F(x, y)=Tr_{m}^{n}(x \pi(y))+G(y),
\end{equation*}
where $G: \mathbb{F}_{p^n}\rightarrow \mathbb{F}_{p^m}$ is a nonzero function satisfying $G(cx)=\theta(c)G(x), c \in \mathbb{F}_{p^m}^{*}, x \in \mathbb{F}_{p^n}$. Then $\{D_{F, i}, i \in \mathbb{F}_{p^m}\}$ is a bent partition belonging to class $\mathcal{WBP}$ that do not correspond to vectorial dual-bent functions.
\end{proposition}

\begin{proof}
Since $\pi(cx)=\theta(c)\pi(x), c \in \mathbb{F}_{p^m}^{*}, x \in \mathbb{F}_{p^n}$, then $\pi^{-1}(cx)=\theta^{-1}(c)\pi^{-1}(x), c \in \mathbb{F}_{p^m}^{*}, x \in \mathbb{F}_{p^n}$. For any $c \in \mathbb{F}_{p^m}^{*}$, $F_{c}(x, y)=Tr_{1}^{n}(cx \pi (y))+Tr_{1}^{m}(cG(y))$. Then by \cite[Theorem 4]{Meidl2022A}, $F_{c}$ is a Maiorana-McFarland bent function with $F_{c}$ is regular and
\begin{equation*}
\begin{split}
(F_{c})^{*}(x, y)&=-Tr_{1}^{n}(\pi^{-1}(c^{-1}x)y)+Tr_{1}^{m}(cG(\pi^{-1}(c^{-1}x)))\\
&=-Tr_{1}^{n}(\theta^{-1}(c^{-1})\pi^{-1}(x)y)+Tr_{1}^{m}(cG(\theta^{-1}(c^{-1})\pi^{-1}(x))).
\end{split}
\end{equation*}
Since $\theta^{-1}(d^{-1})+\theta^{-1}(e^{-1})=\theta^{-1}((d+e)^{-1}), d\neq -e \in \mathbb{F}_{p^m}^{*}$, and $G(cx)=\theta(c)G(x), c \in \mathbb{F}_{p^m}^{*}$, then it is easy to check that for any $d\neq -e \in \mathbb{F}_{p^m}^{*}$,
\begin{equation*}
(F_{d})^{*}(x, y)+(F_{e})^{*}(x, y)-(F_{d+e})^{*}(x, y)=Tr_{1}^{m}(G(\pi^{-1}(x))).
\end{equation*}
Since $\theta$ is a permutation over $\mathbb{F}_{p^m}^{*}$ and $G$ is a nonzero function with $G(cx)=\theta(c)G(x), c \in \mathbb{F}_{p^m}^{*}$, we know that $G$ is surjective, and $Tr_{1}^{m}(G(\pi^{-1}(x)))$ is not the zero function. By Corollary 1, $\{D_{F, i}, i \in \mathbb{F}_{p^m}\}$ is a bent partition belonging to class $\mathcal{WBP}$ that do not correspond to vectorial dual-bent functions.
\end{proof}

\begin{remark}\label{Remark 4}
Some classes of known permutations $\pi$ that satisfying the condition in Proposition 3 are listed in \cite{AFKMWW2025A}. The explicit construction of bent partitions not corresponding to vectorial dual-bent functions given in \cite[Theorem 1]{AFKMWW2025A} can be obtained by Proposition 3 by setting $\pi(x)=x^{-a}, G(x)=M(x^{-b})$, where $a\equiv b\equiv p^{l} \mod (p^m-1), \mathrm{gcd}(a, p^n-1)=\mathrm{gcd}(b,p^n-1)=1$, and $M: \mathbb{F}_{p^n}\rightarrow \mathbb{F}_{p^m}$ is not the zero function with $M(cx)=cM(x), c \in \mathbb{F}_{p^m}^{*}, x \in \mathbb{F}_{p^n}$.
\end{remark}

Below we show that a modified construction of \cite{WFW2023Be} provides a secondary construction of vectorial bent functions with the conditions in Theorem 4, and then more bent partitions not corresponding to vectorial dual-bent functions can be obtained.

\begin{proposition}\label{Proposition 4}
Let $n, n', m$ be positive integers for which $n$ is even, $m\leq \frac{n}{2}, m \mid n', m\neq n'$, and $m\geq 2$ if $p=2$. For any $i \in \mathbb{F}_{p^m}$, let $R^{(i)}: V_{n}^{(p)}\rightarrow \mathbb{F}_{p^m}$ be a vectorial bent function with the following conditions.
\begin{itemize}
  \item Any component function of $R^{(i)}$ is weakly regular with $\varepsilon_{(R^{(i)})_{c}}=\varepsilon, c \in \mathbb{F}_{p^m}^{*}$, where $\varepsilon \in \{\pm 1\}$ is a constant independent of $i$ and $c$.
   \item There are functions $G^{(i)}: V_{n}^{(p)}\rightarrow \mathbb{F}_{p^m}$ and $h^{(i)}: V_{n}^{(p)}\rightarrow \mathbb{F}_{p}$ such that $((R^{(i)})_{c})^{*}(x)=(G^{(i)})_{c}(x)+h^{(i)}(x), c \in \mathbb{F}_{p^m}^{*}, x \in V_{n}^{(p)}$, and $h^{(0)}$ is not the zero function.
\end{itemize}
Let $\alpha, \beta \in \mathbb{F}_{p^{n'}}$ be linearly independent over $\mathbb{F}_{p^{m}}$, and let $P$ be a permutation over $\mathbb{F}_{p^{n'}}$ with $P(0)=0$. Define $F: V_{n}^{(p)} \times \mathbb{F}_{p^{n'}} \times \mathbb{F}_{p^{n'}}\rightarrow \mathbb{F}_{p^m}$ as
\begin{equation*}
F(x, y_{1}, y_{2})=R^{(Tr_{m}^{n'}(\alpha P(y_{1}y_{2}^{-1})))}(x)+Tr_{m}^{n'}(\beta P(y_{1}y_{2}^{-1})).
\end{equation*}
Then $\{D_{F, i}, i \in \mathbb{F}_{p^m}\}$ is a bent partition belonging to class $\mathcal{WBP}$ that do not correspond to vectorial dual-bent functions.
\end{proposition}

\begin{proof}
For any $c \in \mathbb{F}_{p^m}^{*}$ and $(a, b_{1}, b_{2}) \in V_{n}^{(p)} \times \mathbb{F}_{p^{n'}} \times \mathbb{F}_{p^{n'}}$, by the proof of \cite[Theorem 4]{WFW2023Be}, we have
\begin{equation*}
W_{F_{c}}(a, b_{1}, b_{2})=p^{n'}\zeta_{p}^{Tr_{1}^{m}(cTr_{m}^{n'}(\beta P(-b_{1}^{-1}b_{2})))}W_{(R^{(Tr_{m}^{n'}(\alpha P(-b_{1}^{-1}b_{2})))})_{c}}(a).
\end{equation*}
Then $F_{c}$ is weakly regular bent with $\varepsilon_{F_{c}}=\varepsilon$ and
\begin{equation*}
(F_{c})^{*}(x, y_{1}, y_{2})=(G^{(Tr_{m}^{n'}(\alpha P(-y_{1}^{-1}y_{2})))})_{c}(x)+Tr_{1}^{m}(cTr_{m}^{n'}(\beta P(-y_{1}^{-1}y_{2})))+h^{(Tr_{m}^{n'}(\alpha P(-y_{1}^{-1}y_{2})))}(x).
\end{equation*}
Denote
\begin{equation*}
\begin{split}
& U(x, y_{1}, y_{2})=G^{(Tr_{m}^{n'}(\alpha P(-y_{1}^{-1}y_{2})))}(x)+Tr_{m}^{n'}(\beta P(-y_{1}^{-1}y_{2})),
\\ & v(x, y_{1}, y_{2})=h^{(Tr_{m}^{n'}(\alpha P(-y_{1}^{-1}y_{2})))}(x).
\end{split}
\end{equation*}
Then $(F_{c})^{*}(x, y_{1}, y_{2})=U_{c}(x, y_{1}, y_{2})+v(x, y_{1}, y_{2})$. Since $h^{(0)}$ is not the zero function, then $v(x, y_{1}, y_{2})$ is not the zero function. By Theorem 4 and Remark 3, $\{D_{F, i}, i \in \mathbb{F}_{p^m}\}$ is a bent partition not corresponding to vectorial dual-bent functions.
\end{proof}

We give an example to illustrate Proposition 4.

\begin{example} \label{Example 1}
Let $p=3, m=2, n=8, n'=4, \alpha=1, \beta \in \mathbb{F}_{3^4}^{*} \backslash \mathbb{F}_{3^2}^{*}$ and $P$ be the identity map over $\mathbb{F}_{3^4}$. Let $R^{(0)}: \mathbb{F}_{3^4} \times \mathbb{F}_{3^4}\rightarrow \mathbb{F}_{3^2}$ be given by $R^{(0)}(x_{1}, x_{2})=Tr_{2}^{4}(x_{1}^{69}x_{2}+x_{1}^{29})$, and $R^{(i)}: \mathbb{F}_{3^4} \times \mathbb{F}_{3^4}\rightarrow \mathbb{F}_{3^2}, i \in \mathbb{F}_{3^2}^{*}$, be given by $R^{(i)}(x_{1}, x_{2})=Tr_{2}^{4}(x_{1}^{79}x_{2})$. By Proposition 3, $R^{(i)}, i \in \mathbb{F}_{3^2}$, satisfy the conditions in Proposition 4. By Proposition 4, $\{D_{F, i}, i \in \mathbb{F}_{3^2}\}$ is a bent partition of $\mathbb{F}_{3^4} \times \mathbb{F}_{3^4} \times \mathbb{F}_{3^4} \times \mathbb{F}_{3^4}$, where $F(x_{1}, x_{2}, y_{1}, y_{2})={Tr_{2}^{4}(y_{1}y_{2}^{-1})}^{8}Tr_{2}^{4}(x_{1}^{79}x_{2}-x_{1}^{69}x_{2}-x_{1}^{29})+Tr_{2}^{4}(x_{1}^{69}x_{2}+x_{1}^{29}+\beta y_{1}y_{2}^{-1})$ is not a vectorial dual-bent function.
\end{example}

The following theorem gives another secondary construction of vectorial bent functions satisfying the conditions in Theorem 4, and more bent partitions (not) corresponding to vectorial dual-bent functions can be obtained.

\begin{theorem} \label{Theorem 6}
Let $n, n', m, m', r$ be positive integers for which $n, n'$ are even, $m \leq \frac{n}{2}, m'\leq \frac{n'}{2}, r\leq m, r\leq m'$, and $m, m'\geq 2$ if $p=2$. Let $R: V_{n}^{(p)}\rightarrow V_{m}^{(p)}$ and $R': V_{n'}^{(p)}\rightarrow V_{m'}^{(p)}$ be vectorial bent functions satisfying the following conditions.
\begin{itemize}
   \item Any component function of $R$ (respectively, $R'$) is weakly regular with $\varepsilon_{R_{c}}=\varepsilon, c \in {V_{m}^{(p)}}^{*}$ (respectively, $\varepsilon_{R'_{c}}=\varepsilon', c \in {V_{m'}^{(p)}}^{*}$), where $\varepsilon \in \{\pm 1\}$ (respectively, $\varepsilon' \in \{\pm 1\}$) is a constant.
   \item There are functions $G: V_{n}^{(p)}\rightarrow V_{m}^{(p)}$ (respectively, $G': V_{n'}^{(p)}\rightarrow V_{m'}^{(p)}$) and $h: V_{n}^{(p)}\rightarrow \mathbb{F}_{p}$ (respectively, $h': V_{n'}^{(p)}\rightarrow \mathbb{F}_{p}$) such that $(R_{c})^{*}(x)=G_{c}(x)+h(x), c \in {V_{m}^{(p)}}^{*}, x \in V_{n}^{(p)}$ (respectively, $(R'_{c})^{*}(x)=G'_{c}(x)+h'(x), c \in {V_{m'}^{(p)}}^{*}, x \in V_{n'}^{(p)}$).
\end{itemize}
Assume that $K: V_{m}^{(p)} \times V_{m'}^{(p)}\rightarrow V_{r}^{(p)}$ satisfies that for any $x \in V_{m}^{(p)}$, $S^{(x)}: V_{m'}^{(p)}\rightarrow V_{r}^{(p)}$ defined by $S^{(x)}(y)=K(x, y), y \in V_{m'}^{(p)}$, is balanced, and for any $y \in V_{m'}^{(p)}$, $T^{(y)}: V_{m}^{(p)}\rightarrow V_{r}^{(p)}$ defined by $T^{(y)}(x)=K(x, y), x \in V_{m}^{(p)}$, is balanced. Define $F: V_{n}^{(p)} \times V_{n'}^{(p)}\rightarrow V_{r}^{(p)}$ as
\begin{equation*}
F(x, y)=K(R(x), R'(y)).
\end{equation*}
Then $\{D_{F, i}, i \in V_{r}^{(p)}\}$ is a bent partition belonging to class $\mathcal{WBP}$. Furthermore, $F$ is a vectorial dual-bent function if and only if $h$ and $h'$ are both constant functions with $h(x)=-h'(y), x \in V_{n}^{(p)}, y \in V_{n'}^{(p)}$.
\end{theorem}

\begin{proof}
For any $c \in {V_{r}^{(p)}}^{*}$ and $(a, b) \in V_{n}^{(p)} \times V_{n'}^{(p)}$, by Eq. (30), we have
\begin{small}
\begin{equation*}
\begin{split}
&W_{F_{c}}(-a, -b)\\
&=\sum_{x \in V_{n}^{(p)}, y \in V_{n'}^{(p)}}\zeta_{p}^{\langle c, K(R(x), R'(y))\rangle_{r}+\langle a, x\rangle_{n}+\langle b, y\rangle_{n'}}\\
&=\sum_{i \in V_{m}^{(p)}, j \in V_{m'}^{(p)}}\zeta_{p}^{\langle c, K(i, j)\rangle_{r}}\chi_{a}(D_{R, i})\chi_{b}(D_{R', j})\\
&=\sum_{i \in V_{m}^{(p)}, j \in V_{m'}^{(p)}}\zeta_{p}^{\langle c, K(i, j)\rangle_{r}}(p^{n-m}\delta_{0}(a)+\varepsilon p^{\frac{n}{2}-m}\zeta_{p}^{h(-a)}(p^{m}\delta_{G(-a)}(i)-1))\\
& \ \ \ \ \ \ \ \ \ \ \ \ \ \ \ \ \times (p^{n'-m'}\delta_{0}(b)+\varepsilon' p^{\frac{n'}{2}-m'}\zeta_{p}^{h'(-b)}(p^{m'}\delta_{G'(-b)}(j)-1))\\
&=(p^{n-m}\delta_{0}(a)-\varepsilon p^{\frac{n}{2}-m}\zeta_{p}^{h(-a)})\sum_{j \in V_{m'}^{(p)}}(p^{n'-m'}\delta_{0}(b)+\varepsilon' p^{\frac{n'}{2}-m'}\zeta_{p}^{h'(-b)}(p^{m'}\delta_{G'(-b)}(j)-1))\sum_{i \in V_{m}^{(p)}}\zeta_{p}^{\langle c, T^{(j)}(i)\rangle_{r}}\\
& \ \ +\varepsilon p^{\frac{n}{2}}\zeta_{p}^{h(-a)}(p^{n'-m'}\delta_{0}(b)-\varepsilon'p^{\frac{n'}{2}-m'}\zeta_{p}^{h'(-b)})\sum_{j \in V_{m'}^{(p)}}\zeta_{p}^{\langle c, S^{(G(-a))}(j)\rangle_{r}}\\
& \ \ +\varepsilon \varepsilon' p^{\frac{n+n'}{2}}\zeta_{p}^{\langle c, K(G(-a), G'(-b))\rangle_{r}+h(-a)+h'(-b)}.
\end{split}
\end{equation*}
\end{small}For any $i \in V_{m}^{(p)}$ and $j \in V_{m'}^{(p)}$, since $S^{(i)}$ and $T^{(j)}$ are balanced, then
\begin{equation*}
\sum_{x \in V_{m'}^{(p)}}\zeta_{p}^{\langle c, S^{(i)}(x)\rangle_{r}}=\sum_{x \in V_{m}^{(p)}}\zeta_{p}^{\langle c, T^{(j)}(x)\rangle_{r}}=0, c \in {V_{r}^{(p)}}^{*}.
\end{equation*}
Therefore, $W_{F_{c}}(-a, -b)=\varepsilon \varepsilon' p^{\frac{n+n'}{2}}\zeta_{p}^{\langle c, K(G(-a), G'(-b))\rangle_{r}+h(-a)+h'(-b)}$. Note that $h(x)+h'(y), x \in V_{n}^{(p)}, y \in V_{n'}^{(p)}$, is the zero function if and only if $h$ and $h'$ are both constant functions with $h(x)=-h'(y)$. Then the result follows from Theorem 4 and Remark 3.
\end{proof}

\begin{remark} \label{Reamrk 5}
The construction of vectorial Boolean bent functions given in \cite[Proposition 12.1.7]{Mesnager2016Be} can be obtained by Theorem 6 with $R(x_{1}, x_{2})=x_{1}x_{2}^{-1}$ and $R'(y_{1}, y_{2})=y_{1}y_{2}^{-1}$.
\end{remark}

The following corollary is directly obtained by Theorem 6.

\begin{corollary} \label{Corollary 2}
With the same notation as in Theorem 6. Let $m=m'=r$ and $K(x, y)=x+y$, then $\{D_{F, i}, i \in V_{r}^{(p)}\}$ is a bent partition, where $F(x, y)=R(x)+R'(y)$.
\end{corollary}

If $P$ is a permutation of $\mathbb{F}_{p^r}$ and $P(x)+x$ is also a permutation of $\mathbb{F}_{p^r}$, then $P$ is called a \emph{complete permutation}. For complete permutation, please refer to \cite{Hou2015Pe} and references therein. By using complete permutation, we obtain the following corollary.

\begin{corollary} \label{Corollary 3}
With the same notation as in Theorem 6. Let $m=m'=r, V_{m}^{(p)}=V_{m'}^{(p)}=V_{r}^{(p)}=\mathbb{F}_{p^r}$ and $K(x, y)=P(x+y)+x$, where $P$ is a complete permutation of $\mathbb{F}_{p^r}$. Then $\{D_{F, i}, i \in \mathbb{F}_{p^r}\}$ is a bent partition, where $F(x, y)=P(R(x)+R'(y))+R(x)$.
\end{corollary}

\begin{proof}
We only need to prove that $K$ satisfies the conditions in Theorem 6. Since $P$ is a complete permutation, then $P(x)$ and $P(x)+x$ are permutations of $\mathbb{F}_{p^r}$. As $P$ is a permutation of $\mathbb{F}_{p^r}$, then for any $x \in \mathbb{F}_{p^r}$, $y\mapsto P(x+y)+x$ is a permutation of $\mathbb{F}_{p^r}$. As $P(x)+x$ is a permutation of $\mathbb{F}_{p^r}$, then for any $y \in \mathbb{F}_{p^r}$, $x\mapsto P(x+y)+x+y$ is a permutation of $\mathbb{F}_{p^r}$, and thus $x\mapsto P(x+y)+x$ is a permutation of $\mathbb{F}_{p^r}$.
\end{proof}

We give examples to illustrate Corollaries 2 and 3.

\begin{example} \label{Example 2}
Let $R: \mathbb{F}_{3^4} \times \mathbb{F}_{3^4}\rightarrow \mathbb{F}_{3^2}$ be given by $R(x_{1}, x_{2})=Tr_{2}^{4}(x_{1}^{69}x_{2}+x_{1}^{29})$, and $R': \mathbb{F}_{3^4} \times \mathbb{F}_{3^4}\rightarrow \mathbb{F}_{3^2}$ be given by $R'(y_{1}, y_{2})=Tr_{2}^{4}(y_{1}^{79}y_{2})$. By Proposition 3, $R$ and $R'$ satisfy the conditions in Corollary 2. By Corollary 2, $\{D_{F, i}, i \in \mathbb{F}_{3^2}\}$ is a bent partition of $\mathbb{F}_{3^4} \times \mathbb{F}_{3^4} \times \mathbb{F}_{3^4} \times \mathbb{F}_{3^4}$, where $F(x_{1}, x_{2}, y_{1}, y_{2})=Tr_{2}^{4}(x_{1}^{69}x_{2}+x_{1}^{29}+y_{1}^{79}y_{2})$ is not a vectorial dual-bent function.
\end{example}

\begin{example} \label{Example 3}
Let $R, R': \mathbb{F}_{2^6} \times \mathbb{F}_{2^6}\rightarrow \mathbb{F}_{2^6}$ be given by $R(x_{1}, x_{2})=x_{1}^{-1}x_{2}$, $R'(y_{1}, y_{2})=y_{1}y_{2}^{-1}$. By Proposition 3, $R$ and $R'$ satisfy the conditions in Corollary 3. Let $P(x)=x^{17}+x^{5}+\alpha x, x \in \mathbb{F}_{2^6}$, where $\alpha \in \mathbb{F}_{2^2} \backslash \mathbb{F}_{2}$. Then by \cite[Page 99]{Hou2015Pe}, $P$ is a complete permutation of $\mathbb{F}_{2^6}$. By Corollary 3, $\{D_{F, i}, i \in \mathbb{F}_{2^6}\}$ is a bent partition of $\mathbb{F}_{2^6} \times \mathbb{F}_{2^6} \times \mathbb{F}_{2^6} \times \mathbb{F}_{2^6}$, where $F(x_{1}, x_{2}, y_{1}, y_{2})=P(x_{1}^{-1}x_{2}+y_{1}y_{2}^{-1})+x_{1}^{-1}x_{2}=(x_{1}^{-1}x_{2}+y_{1}y_{2}^{-1})^{17}+(x_{1}^{-1}x_{2}+y_{1}y_{2}^{-1})^{5}
+(\alpha+1)x_{1}^{-1}x_{2}+\alpha y_{1}y_{2}^{-1}$ is a vectorial dual-bent function.
\end{example}

\section{Conclusion}
In this paper, bent partitions were further studied. For a bent partition of $V_{n}^{(p)}$ for which all generated $p$-ary bent functions are regular or all are weakly regular but not regular, we proved that its depth must be a power of $p$. For any characteristic $p$, we give a sufficient condition for the preimage set partition of a vectorial bent function to be a bent partition, based on which new constructions of bent partitions (not) corresponding to vectorial dual-bent functions were presented. Particularly, we provided a new construction of vectorial dual-bent functions. For $p=2$, we characterized general bent partitions in terms of Hadamard matrices.

To conclude this paper, some directions for further research are listed:
\begin{itemize}
  \item Regrading the depth problem of bent partitions, we have solved the case for those belonging to class $\mathcal{WBP}$. To fully resolve this challenging problem, future work should focus on the existence of bent partitions generating non-weakly regular bent functions, or bent partitions simultaneously generating regular and weakly regular but not regular bent functions.
  \item Theorem 4 provided a sufficient condition for constructing bent partitions belonging to class $\mathcal{WBP}$. It is interesting to investigate whether the sufficient condition is also necessary (see Remark 2).
  \item Based on Theorem 4 and Corollary 1, how to construct more bent partitions.
\end{itemize}

\end{document}